\documentclass[11pt]{article}
\usepackage{amsmath, amsthm, amssymb, bbm, multirow, graphicx, enumitem, mathrsfs, dsfont}
\usepackage[margin=1in]{geometry}

\theoremstyle{plain}
\newtheorem{theorem}{Theorem}[section]
\newtheorem{lemma}[theorem]{Lemma}

\newtheorem{conjecture}[theorem]{Conjecture}

\theoremstyle{definition}

\theoremstyle{remark}

\newtheorem{remark}[theorem]{Remark}

\newcommand{\cS}{\mathcal{S}}
\newcommand{\cB}{\mathcal{B}}
\newcommand{\cA}{\mathcal{A}}

\newcommand{\bN}{\mathbb{N}}
\newcommand{\bcS}{\boldsymbol{\mathcal{S}}}
\newcommand{\bcA}{\boldsymbol{\mathcal{A}}}
\newcommand{\bcP}{\boldsymbol{\mathcal{P}}}
\newcommand{\ty}{\tilde{y}}
\newcommand{\tx}{\tilde{x}}
\newcommand{\bb}{\mathds{1}}
\newcommand{\timesf}{\mathsf{time}}
\newcommand{\depth}{\mathsf{depth}}
\newcommand{\Par}{\mathsf{Par}}
\newcommand{\Free}{\mathsf{Free}}
\newcommand{\Freeemp}{\mathsf{Free}^{\mathsf{emp}}}
\newcommand{\OPT}{\mathsf{OPT}}
\newcommand{\current}{\mathsf{current}}
\newcommand{\med}{\mu_{\varepsilon,\delta}}
\newcommand{\started}{\mathsf{Started}}
\newcommand{\played}{\mathsf{Played}}
\newcommand{\subr}{\mathsf{Policy}}

\newcommand{\Halmos}{}
\newcommand{\up}{}
\newcommand{\down}{}

\title{Improvements and Generalizations of Stochastic Knapsack and Markovian Bandits Approximation Algorithms}
\author{
  Will Ma\footnote{Operations Research Center, Massachusetts Institute of Technology, \texttt{willma@mit.edu}.}
}

\begin{document}

\maketitle

\begin{abstract}
We study the multi-armed bandit problem with arms which are Markov chains with rewards.  In the finite-horizon setting, the celebrated Gittins indices do not apply, and the exact solution is intractable.  We provide approximation algorithms for a more general model which includes Markov decision processes and non-unit transition times.  When preemption is allowed, we provide a $(\frac{1}{2}-\varepsilon)$-approximation, along with an example showing this is tight.  When preemption isn't allowed, we provide a $\frac{1}{12}$-approximation, which improves to a $\frac{4}{27}$-approximation when transition times are unity.  Our model encompasses the Markovian Bandits model of Gupta et al, the Stochastic Knapsack model of Dean, Goemans, and Vondrak, and the Budgeted Learning model of Guha and Munagala, and our algorithms improve existing results in all three areas.  In our analysis, we encounter and overcome to our knowledge a novel obstacle---an algorithm that provably exists via polyhedral arguments, but cannot be found in polynomial time.
\end{abstract}

\clearpage

\section{Introduction.} We are interested in a broad class of stochastic control problems: there are multiple evolving systems competing for the attention of a single operator, who has limited time to extract as much reward as possible.  Classical examples include a medical researcher allocating his time between different clinical trials, or a graduate student shifting her efforts between different ongoing projects.  Before we describe our model in detail, we introduce the three problems in the literature which are special cases of our problem, and motivated our avenues of generalization.

\subsection{Markovian Bandits.} The Markovian multi-armed bandit problem is the following: there are some number of Markov chains (\emph{arms}), each of which only evolve to the next node\footnote{We use the word \emph{node} instead of \emph{state} to avoid confusion with the notion of a state in dynamic programming.} and return some reward when you play (\emph{pull}) that arm; the controller has to allocate a fixed number of pulls among the arms to maximize expected reward.  The reward returned by the next pull of an arm depends on the current node that arm is on.  When an arm is pulled, the controller observes the transition taken before having to choose the next arm to pull.  Multi-armed bandit (MAB) problems capture the tradeoff between \emph{exploring} arms that could potentially transition to high-reward nodes, versus \emph{exploiting} arms that have the greatest immediate payoff.

The infinite-horizon version of this problem with discounted rewards can be solved by the celebrated index policy of Gittins; see the book \cite{GGW11} for an in-depth treatment of Gittins indices.  However, all of this theory is crucially dependent on the time horizon being infinite (see \cite[sect.~3.4.1]{GGW11}).  The Gittins index measures the asymptotic performance of an arm, and does not apply when there is a discrete number of time steps remaining.

Also, when we refer to multi-armed bandit in this paper, it is not to be confused with the popular Stochastic Bandits model, where each arm is an unknown reward distribution, playing that arm collects a random sample from its distribution, and the objective is to learn which arm has the highest mean in a way that minimizes regret.  For a comprehensive summary on Stochastic Bandits and related bandit models, we refer the reader to the survey of Bubeck and Cesa-Bianchi \cite{BCB12}.  The main difference with our problem is that despite its stochastic nature, all of the transition probabilities are given as input and we can define an exact optimization problem (and the challenge is computational), whereas in Stochastic Bandits there is uncertainty in the parameter information (and the challenge is to compete with an omniscient adversary).

The finite-horizon Markovian Bandits problem is intractable even in special cases (see Goel et al.\ \cite{GGM06}, and the introduction of Guha and Munagala \cite{GM13}), so we turn to approximation algorithms.  The state of the art is an LP-relative $\frac{1}{48}$-approximation\footnote{All of the problems we discuss will be maximization problems, for which an $\alpha$-approximation refers to an algorithm that attains at least $\alpha$ of the optimum.} by Gupta et al.\ \cite{GKMR11}.  Our results improve this bound by providing an LP-relative $\frac{4}{27}$-approximation for a more general problem.

\subsubsection{Martingale Reward Bandits and Bayesian Bandits.} While Markovian Bandits is a different problem from Stochastic Bandits, it is a generalization of the closely related Bayesian Bandits, where each arm is an unknown reward distribution, but we have prior beliefs about what these distributions may be, and we update our beliefs as we collect samples from the arms.  The objective is to maximize expected reward under a fixed budget of plays.

For each arm, every potential posterior distribution can be represented by a node in a Markov chain, and the transitions between nodes correspond to the laws of Bayesian inference.  However, the resulting Markov chain is forced to satisfy the \emph{martingale condition}, ie.\ the expected reward at the next node must equal the expected reward at the current node, by Bayes' law.  This condition is not satisfied by Stochastic Knapsack with correlated rewards, as well as certain natural applications of the bandit model.  For instance, in the marketing problems studied by Bertsimas and Mersereau \cite{BM07}, the arms represent customers who may require repeated pulls (marketing actions) before they transition to a reward-generating node.

Nonetheless, fruitful research has been done in the Bayesian Bandits setting --- Guha and Munagala \cite{GM13} show that constant-factor approximations can be obtained even under a variety of side constraints.  The complexity necessary for a policy to be within a constant factor of optimal is much lower under the martingale assumption.  For the basic bandit problem with no side constraints, Farias and Madan \cite{FM11} observe that \emph{irrevocable} policies --- policies which cannot start an arm, stop pulling it at some point, and resume it later --- extract a constant fraction of the optimal (non-irrevocable) reward.  Motivated by this, \cite{GM13} obtains a $(\frac{1}{2}-\varepsilon)$-approximation for Bayesian Bandits that is in fact a irrevocable policy.

\subsubsection{Irrevocable Bandits.} The above can be contrasted with the work of Gupta et al., who construct a non-martingale instance where irrevocable policies (they refer to these policies as \emph{non-preempting}) can only extract an arbitrarily small fraction of the optimal reward \cite[appx.~A.3]{GKMR11}.  Therefore, without the martingale assumption, we can only hope to compare irrevocable policies against the irrevocable optimum.  We provide a $(\frac{1}{2}-\varepsilon)$-approximation for this problem, which we refer to as \emph{Irrevocable Bandits}.

\subsection{Stochastic Knapsack.} The Stochastic Knapsack (SK) problem was introduced by Dean et al.\ in 2004 \cite{DGV04} (see \cite{DGV08} for the journal version).  We are to schedule some jobs under a fixed time budget.  Each job has a stochastic reward and processing time whose distribution is known beforehand.  We sequentially choose which job to perform next, only discovering its length and reward in real-time as it is being processed.  The objective is to maximize the expected reward before the time budget is spent.  A major focus of their work is on the benefit of \emph{adaptive} policies (which can make dynamic choices based on the instantiated lengths of jobs processed so far) over \emph{non-adaptive} policies (which must fix an ordering of the jobs beforehand), but in our work all policies will be adaptive.

Throughout \cite{DGV08}, the authors assume uncorrelated rewards --- that is, the reward of a job is independent of its length.  The state of the art for this setting is a $(\frac{1}{2}-\varepsilon)$-approximation by Bhalgat \cite{Bha11}; a $(\frac{1}{2}-\varepsilon)$-approximation is also obtained for the variant where jobs can be canceled at any time by Li and Yuan \cite{LY13}.  \cite{GKMR11} provides a $\frac{1}{8}$-approximation for Stochastic Knapsack with potentially correlated rewards, and a $\frac{1}{16}$-approximation for the variant with cancellation.  We improve these bounds by providing an LP-relative $(\frac{1}{2}-\varepsilon)$-approximation for a problem which generalizes both variants with correlated rewards.  Furthermore, we construct an example where the true optimum is as small as $\frac{1}{2}+\varepsilon$ of the optimum of the LP relaxation.  Therefore, our bound is tight in the sense that one cannot hope to improve the approximation ratio using the same LP relaxation.

However, it is important to mention that our results, as well as the results of \cite{GKMR11}, require the job sizes and budget to be given in unary, since these algorithms use a time-indexed LP.  It appears that this LP is necessary whenever correlation is allowed --- the non-time-indexed LP can be off by an arbitrarily large factor (see \cite[appx.~A.2]{GKMR11}).  Techniques for discretizing the time-indexed LP if the job sizes and budget are given in binary are provided in \cite{GKMR11}, albeit losing some approximation factor.  Nonetheless, in this paper, we always think of processing times as discrete hops on a Markov chain, given in unary.  Note that stronger hardness results than those aforementioned can be obtained when the sizes are given in binary (see Dean et al.\ \cite{DGV05}, and the introduction of \cite{DGV08}).

\subsection{Futuristic Bandits and Budgeted Bandits.} Starting with \cite{GM07a,GM07b}, Guha and Munagala have studied many variants of \emph{budgeted learning} problems --- including switching costs, concave utilities, and Lagrangian budget constraints.  See \cite{GM08} for an updated article that also subsumes some of their other works.  Their basic setting, which we refer to as \emph{Futuristic Bandits}, is identical to Bayesian Bandits (ie.\ there are Markov chains satisfying the martingale condition), except no rewards are dispensed during the execution of the algorithm.  Instead, once the budget\footnote{In some variants, there is a cost budget instead of a time budget, and exploring each arm incurs a different cost.  We explain in Section~\ref{preliminaries} why our model also generalizes this setting, which they refer to as \emph{Budgeted Bandits}.} is spent, we pick the arm we believe to be best, and only earn the (expected) reward for that arm.  A $\frac{1}{4}$-approximation is provided in \cite{GM08}, and this is improved by the same authors to a $(\frac{1}{3}-\varepsilon)$-approximation in \cite{GM13}.  Our algorithm works without the martingale assumption, but the approximation guarantee is only $\frac{4}{27}$.

\subsection{MAB superprocess with multi-period actions.}\label{MABSMA}

Motivated by these examples, we now introduce our generalized problem, which we call \emph{MAB superprocess with multi-period actions}.  Consider the Markovian Bandits setting, except we allow for a more general family of inputs, in two ways.

First, we allow transitions on the Markov chains to consume more than one pull worth of budget.  We can think of these transitions as having a non-unit \emph{processing time}.  The processing times can be stochastic, and correlated with the node transition that takes place.  The rewards can be accrued upon pulling the node, or only accrued if the processing time completes before the time budget runs out.  The applications of such a generalization to US Air Force jet maintenance have recently been considered in Kessler's thesis \cite{Kes13}, where it is referred to as \emph{multi-period actions}.

The second generalization is that we allow each arm to be a Markov decision process; such a problem is referred to as \emph{MAB superprocess} in Gittins et al.\ \cite{GGW11}.  Now, when the controller pulls an arm, they have a choice of actions, each of which results in a different joint distribution on reward, processing time, and transition taken.

The purpose of the first generalization is to allow MAB to model the jobs from Stochastic Knapsack which have rewards correlated with processing time and can't be canceled.  The purpose of the second generalization is to allow MAB to model Futuristic Bandits, where exploiting an arm corresponds to a separate action.  The details of our reductions, along with examples, will be presented throughout Section~\ref{preliminaries}, once we have introduced formal notation.

We consider two variants under our general setting: the case with preemption (ie.\ we can start playing an arm, not play it for some time steps, and resume playing it later), and the case without preemption.  The variant without preemption is necessary to generalize Stochastic Knapsack and Irrevocable Bandits.  The variant with preemption generalizes Markovian Bandits and Futuristic Bandits.

\subsection{Outline of results.}

Our main results can be outlined as follows:
\begin{itemize}
	\item Reductions from earlier problems to \emph{MAB superprocess with multi-period actions} [sect.~\ref{preliminaries}]
	\item Polynomial-sized LP relaxations for both variants of \emph{MAB superprocess with multi-period actions}, and polyhedral proofs that they are indeed relaxations [sect.~\ref{lp_relaxation_section}]
	\item A $(\frac{1}{2}-\varepsilon)$-approximation for \emph{MAB superprocess with multi-period actions---no preemption}, with runtime polynomial in the input and $\frac{1}{\varepsilon}$ [sect.~\ref{proof_section1}]
	\item A matching upper bound where it is impossible to obtain more than $\frac{1}{2}+\varepsilon$ of the optimum of the LP relaxation [sect.~\ref{upper_bound}]
	\item A $\frac{4}{27}$-approximation for \emph{MAB superprocess} (with preemption) [sect.~\ref{proof_section2}]
	\item A $\frac{1}{12}$-approximation for \emph{MAB superprocess with multi-period actions} (and preemption) [sect.~\ref{proof_section2.5}]
\end{itemize}

The way in which these approximation ratios improve previous results on SK and MAB is summarized in Tables~\ref{table1} and \ref{table2}\footnote{Some of these results have appeared in a preliminary conference version of this article \cite{Ma14}.}.

\begin{table}
	\caption{Comparison of results for SK.\label{table1}}
	\center{
		\begin{tabular}{|c||c|c|}
			\hline
			\up & Previous & Result as a Special \\
			\down Problem & Result & Case of Our Problems \\
			\hline
			\up\down Binary SK & $\frac{1}{2}-\varepsilon$ \cite{Bha11} & - \\
			\up\down Binary SK w/ Cancellation & $\frac{1}{2}-\varepsilon$ \cite{LY13} & - \\
			\up\down Unary Correlated SK & $\frac{1}{8}$ \cite{GKMR11} & $\frac{1}{2}-\varepsilon$ [thm.~\ref{mainresult1}], \cite{Ma14} \\
			\up\down Unary Correlated SK w/ Cancellation & $\frac{1}{16}$ \cite{GKMR11} & $\frac{1}{2}-\varepsilon$ [thm.~\ref{mainresult1}] \\
			\hline
		\end{tabular}
	}
\end{table}

\begin{table}
	\caption{Comparison of results for MAB.\label{table2}}
	\center{
		\begin{tabular}{|c||c|c|c|}
			\hline
			\up & Previous & Result as a Special & Result with \\
			\down Problem & Result & Case of Our Problems & Martingale Assumption \\
			\hline
			\up\down Markovian Bandits & $\frac{1}{48}$ & $\frac{4}{27}$ [thm.~\ref{mainresult2}], \cite{Ma14} & $\frac{1}{2}-\varepsilon$ \cite{GM13} \\
			\up\down Irrevocable Bandits & - & $\frac{1}{2}-\varepsilon$ [thm.~\ref{mainresult1}] & $\frac{1}{2}-\varepsilon$ \cite{GM13} \\
			\up\down Futuristic Bandits & - & $\frac{4}{27}$ [thm.~\ref{mainresult2}] & $\frac{1}{3}-\varepsilon$ \cite{GM13} \\
			\up\down Budgeted Bandits & - & $\frac{1}{12}$ [thm.~\ref{mainresult3}] & $\frac{1}{4}-\varepsilon$ \cite{GM08} \\
			\hline
		\end{tabular}
	}
\end{table}

\subsection{Sketch of techniques.}

In the variant without preemption, we show that given any feasible solution to the LP relaxation, there exists a policy which plays every node with half the probability it is played in the LP solution.  This would yield a $\frac{1}{2}$-approximation, but the policy cannot be specified in polynomial time, because the previous argument is purely existential.  Instead, we show how to approximate the policy via sampling, in a way that doesn't cause error propagation.

In the variant with preemption, we provide a priority-based approximation algorithm which intuitively achieves the same goals as the algorithm of Gupta et al.\ \cite{GKMR11}.  However, our algorithm, which has eliminated the need for their \emph{convex decomposition} and \emph{gap filling} operations, allows for a tighter analysis.  Furthermore, it grants the observation that we can get the same result for the more general model of Markov decision processes and non-unit transition times.  Finally, our analysis makes use of Samuels' conjecture \cite{Sam66} for $n=3$ (which is proven), and this helps us bound the upper tail.

\subsection{Related work.}

Most of the related work on bandits, stochastic knapsack/packing, and budgeted learning have already been introduced in the earlier subsections, so we only mention the neglected results here.  One such result for stochastic knapsack is the bi-criteria $(1-\varepsilon)$-approximation of Bhalgat et al.\ \cite{BGK11} that uses $1+\varepsilon$ as much time; such a result is also obtained via alternate methods by Li and Yuan \cite{LY13} and generalized to the setting with both correlated rewards and cancellation.  A new model that adds geometry to SK by associating jobs with locations in a metric space is the \emph{stochastic orienteering} problem introduced by Gupta et al.\ \cite{GKNR14}.  The benefit of adaptive policies for this problem is also addressed by Bansal and Nagarajan \cite{BN14}.

Another example of a stochastic optimization problem where adaptive policies are necessary is the \emph{stochastic matching} problem of Bansal et al.\ \cite{BGL+12} --- in fact we use one of their lemmas in our analysis.  Recently, the setting of stochastic matching has been integrated into online matching problems by Mehta and Panigrahi \cite{MP12}.

All of the problems described thus far deal only with expected reward.  Recently, Ilhan et al.\ \cite{IID11} studied the variant of SK where the objective is to maximize the probability of achieving a target reward; older work on this model includes Carraway et al. \cite{CSW93}.  Approximation algorithms for minimizing the expected sum of weighted completion times when the processing times are stochastic are provided in M{\"o}hring et al.\ \cite{MSU99}, and Skutella and Uetz \cite{SU01}.  SK with chance constraints --- maximizing the expected reward subject to the probability of running overtime being at most $p$ --- is studied in Goel and Indyk \cite{GI99}, and Kleinberg et al.\ \cite{KRT00}.

Looking at more comprehensive synopses, we point the reader interested in infinite-horizon Markovian Bandits to the book by Gittins et al.\ \cite{GGW11}.  Families of bandit problems other than Markovian, including Stochastic and Adversarial, are surveyed by Bubeck and Cesa-Bianchi \cite{BCB12}.  For an encyclopedic treatment of using dynamic programming to solve stochastic control problems, we refer the reader to the book by Bertsekas \cite{Ber95}; for stochastic scheduling in particular, we refer the reader to the book by Pinedo \cite{Pin12}.

\section{Fully generalized model.}\label{preliminaries}

We set out to define notation for the \emph{MAB superprocess with multi-period actions} problem described in Subsection~\ref{MABSMA}.  Let $n\in\bN$ denote the number of arms, which are Markov decision processes with rewards.  There is a budget of $B\in\bN$ time steps over which we would like to extract as much reward in expectation as possible.  The arms could also have multi-period actions, which are transitions that take more than one time step to complete.

Formally, for each arm $i$, let $\cS_i$ denote its finite set of nodes, with the root node being $\rho_i$.  To \emph{play} an arm $i$ that is currently on node $u\in\cS_i$, we select an action $a$ from the finite, non-empty action set $A$, after which the arm will transition to a new node $v\in\cS_i$ in $t$ time steps, accruing reward over this duration.  We will also refer to this process as \emph{playing action} $a$ \emph{on node} $u$, since for each pair $(u,a)$, we are given as input the joint distribution of the destination node, transition time, and reward.  Specifically, for all $a\in A$, $u,v\in\cS_i$, and $t\in[B]$\footnote{For any positive integer $m$, $[m]$ refers to the set $\{1,\ldots,m\}$.}, let $p^a_{u,v,t}$ denote the probability of transitioning to node $v$ in exactly $t$ time steps, when action $a$ is played on node $u$.  We will refer to this transition by the quadruple $(u,a,v,t)$, and when it occurs, let $r^a_{u,v,t,t'}\in[0,\infty)$ denote the reward accrued $t'$ time steps from the present, for all $t'=0,\ldots,t-1$.  We will impose that $\sum_{v\in\cS_i}\sum_{t=1}^Bp^a_{u,v,t}=1$ for all $i\in[n]$, $u\in\cS_i$, and $a\in A$.

\subsection{Simple reductions.}\label{simple_reductions}

We would like to explain why some of the presumptions in the preceding paragraph are WOLOG:
\begin{itemize}
	\item We assumed that all nodes across all arms have the same set $A$ of feasible actions.  This can be easily achieved by taking unions of action sets, and defining inadmissible actions to be duplicates of admissible ones.
	\item We assumed that all transitions take time $t\le B$, with $B$ being the time budget.  Indeed, all transitions $(u,a,v,t)$ with $t>B$ can be amalgamated into the transition $(u,a,v,B)$, since the exact value of $t$ will never be relevant, and all of the rewards $r^a_{u,v,t,t'}$ for $t'\ge B$ will never be accrued.
	\item We assumed that each reward $r^a_{u,v,t,t'}$ is deterministic.  Had they been random with known distributions instead, our problem does not change if we replace each reward with its certainty equivalent, since the objective value only cares about expected reward.  Furthermore, even if the rewards over $t'=0,\ldots,t-1$ were correlated, ie.\ a low reward at $t'=0$ could warn us of lower rewards later on, we cannot interrupt the transition once it begins, so this information is irrelevant.  Note that we do still need to split up the expected reward over $t'=0,\ldots,t$, since our time budget $B$ may run out in the middle of the transition.
	\item We assumed that $\sum_{v\in\cS_i}\sum_{t=1}^Bp^a_{u,v,t}=1$, ie.\ the arm always transitions onto a new node instead of stopping.  This clearly loses no generality since we can always add unit-time self-loop transitions with zero reward.
\end{itemize}

Next, we perform a more complicated reduction to eliminate all transitions with non-unit processing times.  For any transition $(u,a,v,t)$ with $t>1$:
\begin{enumerate}
	\item Add dummy nodes $w_1,\ldots,w_{t-1}$.
	\item Set transition probability $p^a_{u,w_1,1}=p^a_{u,v,t}$, and change $p^a_{u,v,t}$ to be $0$.
	\item Set transition probabilities $p^b_{w_1,w_2,1}=\ldots=p^b_{w_{t-1},v,1}=1$ for all $b\in A$.
	\item Set all other transition probabilities involving $w_1,\ldots,w_{t-1}$ to be $0$.
	\item Set rewards $r^a_{u,w_1,1,0}=r^a_{u,v,t,0}$, $r^b_{w_1,w_2,1,0}=r^a_{u,v,t,1},\ldots,r^b_{w_{t-1},v,1,0}=r^a_{u,v,t,t-1}$ for all $b\in A$.
\end{enumerate}
We call $w_1,\ldots,w_{t-1}$ \emph{bridge} nodes.  So long as we enforce the bridge nodes must be played as soon as they are reached, it is clear that the new problem is equivalent to the old problem.  Repeat this process over all transitions $(u,a,v,t)$ with $t>1$.  We will eliminate the subscripts $t,t'$ and just write $p^a_{u,v}$, $r^a_{u,v}$ now that all transitions with $t>1$ have been reduced to occur with probability $0$.

It will also be convenient to eliminate the $v$ subscript from $r^a_{u,v}$ now that there are no more processing times.  For all $i\in[n]$, $u\in\cS_i$, and $a\in A$, define $r^a_u=\sum_{v\in\cS_i}p^a_{u,v}\cdot r^a_{u,v}$, and consider the Markov decision process that earns deterministic reward $r^a_u$ every time action $a$ is played on node $u$, instead of a random reward $r^a_{u,v}$ that depends on the destination node $v$.  Under the objective of maximizing expected reward, the two models are equivalent, as explained in the third bullet above.

Finally, we assume we have converted each rooted Markov decision process into a layered acyclic digraph, up to depth $B$.  That is, there exists a function $\depth$ mapping nodes to $0,\ldots,B$ such that $\depth(\rho_i)=0$ for all $i\in[n]$, and all transitions $(u,a,v)$ with $p^a_{u,v}>0$ satisfy $\depth(v)=\depth(u)+1$.  This can be done by expanding each node in the original graph into a time-indexed copy of itself for $t=1,\ldots,B$---we refer to \cite[appx.~E.1]{GKMR11} for the standard reduction, which immediately generalizes to the case of Markov decision processes with bridge nodes.

\begin{remark}\label{computational_overhead}
	We would like to point out that our reductions and transformations are stated in a way to maximize ease of exposition.  While they are all polynomial-time, it would reduce computational overhead to remove irrelevant nodes before implementation.
\end{remark}

\subsection{Problem statement.} After all the reductions, let $\cS_i$ denote the set of nodes of arm $i$, and let $\cS=\bigcup_{i=1}^n\cS_i$.  Let $\cB\subset\cS$ denote the set of bridge nodes; note that $\rho_i\in\cB$ for any $i$ is not possible.  Let $\Par(u)=\{(v,a)\in\cS\times A:p^a_{v,u}>0\}$, the (node, action) combinations that have a positive probability of transitioning to $u$.

Each Markov decision process $i$ starts on its root node, $\rho_i$.  At each time step, we choose an arm to play along with an action $a$, getting reward $r^a_u$, where $u$ is the node that arm was on.  We realize the transition that is taken before deciding the arm and action for the next time step.  Of course, if $u$ transitions onto a bridge node, then we have no decision in the next time step, being forced to play the same arm again, say with a default action $\alpha\in A$.  For convenience, we will also allow ourselves to play no arm at a time step\footnote{Since rewards are non-negative, doing nothing cannot be optimal.  However, it makes the analysis cleaner.}.  The objective is to maximize the expected reward accrued after a budget of $B$ time steps.

Algorithms for this problem are described in the form of an adaptive \emph{policy}, a specification of which arm and action to play for each state the system could potentially be in.  A state in this case is determined by the following information: the node each arm is on, and the time step we are at\footnote{Even though we have converted all Markov decision processes into layered acyclic digraphs, we cannot deduce the time elapsed from the nodes each arm is on, since we allow ourselves to not play any arm at a time step.  Therefore, the time step must be included separately in the state information.}.  The optimal policy could theoretically be obtained by dynamic programming, but of course there are exponentially many states, so this is impractical.

However, we still write the Bellman state-updating equations as constraints to get a linear program whose feasible region is precisely the set of admissible policies.  After adding in the objective function of maximizing expected reward, solving this exponential-sized linear program would be equivalent to solving our problem to optimality.

First we need a bit more notation.  Let $\bcS=\cS_1\times\ldots\times\cS_n$; we call its elements \emph{joint nodes}.  For $\pi\in\bcS$ and $u\in\cS_i$, let $\pi^u$ denote the joint node where the $i$'th component of $\pi$ has been replaced by $u$.  A state can be defined by a joint node $\pi$ with a time $t$.  Let $y_{\pi,t}$ be the probability of having arms on nodes according to $\pi$ at the beginning of time $t$.  Let $z^a_{\pi,i,t}$ be the probability we play arm $i$ at time $t$ with action $a$, when the arms are on nodes according to $\pi$.  Some $(\pi,t)$ pairs are impossible states; for example we could never be at a joint node with two or more arms on bridge nodes, and we could not get an arm on a node of depth $5$ at the beginning of time $5$.  However, we still have variables for these $(\pi,t)$ pairs\footnote{Once again, we prioritize notational convenience over computational conciseness; see Remark~\ref{computational_overhead}.}.

Our objective is
\begin{equation}
\label{E0}
\max\sum_{\pi\in\bcS}\sum_{i=1}^n\sum_{a\in A}r^a_{\pi_i}\sum_{t=1}^B z^a_{\pi,i,t}
\end{equation}
with the following constraints on how we can play the arms:
\begin{subequations}
	\begin{align}
	\sum_{i=1}^n\sum_{a\in A}z^a_{\pi,i,t} & \le y_{\pi,t} && \pi\in\bcS,\ t\in[B] \label{E2a} \\
	z^{\alpha}_{\pi,i,t} & =y_{\pi,t} && \pi\in\bcS,\ i:\pi_i\in\cB,\ t\in[B] \label{E2b} \\
	z^a_{\pi,i,t} & \ge0 && \pi\in\bcS,\ i\in[n],\ a\in A,\ t\in[B] \label{E3}
	\end{align}
\end{subequations}
The novel constraint is (\ref{E2b}), which guarantees that we must play a bridge node upon arrival.  The remaining constraints update the $y_{\pi,t}$'s correctly:
\begin{subequations}
	\begin{align}
	y_{(\rho_1,\ldots,\rho_n),1} & =1 && \label{E4a} \\
	y_{\pi,1} & =0 && \pi\in\bcS\setminus\{(\rho_1,\ldots,\rho_n)\} \label{E4b} \\
	y_{\pi,t} & = y_{\pi,t-1}-\sum_{i=1}^n\sum_{a\in A}z^a_{\pi,i,t-1}+\sum_{i=1}^n\sum_{(u,a)\in\Par(\pi_i)}z^a_{\pi^u,i,t-1}\cdot p^a_{u,\pi_i} && t>1,\ \pi\in\bcS \label{E5}
	\end{align}
\end{subequations}

Essentially, the only decision variables are the $z$-variables; there are as many $y$-variables as equalities in (\ref{E4a})-(\ref{E5}).  These constraints guarantee $\sum_{\pi\in\bcS}y_{\pi,t}=1$ for all $t\in[B]$, and combined with (\ref{E2a}), we obtain
\begin{align}
\sum_{\pi\in\bcS}\sum_{i=1}^n\sum_{a\in A}z^a_{\pi,i,t} & \le1 && t\in[B] \label{E1}
\end{align}

Let $\mathtt{(ExpLP)}$ denote the linear program defined by objective (\ref{E0}) and constraints (\ref{E2a})-(\ref{E3}), (\ref{E4a})-(\ref{E5}) which imply (\ref{E1}).  This formally defines our problem, which we will call \emph{MAB superprocess with multi-period actions}.

\subsection{No preemption variant.} We describe the variant of the problem where preemption is not allowed.  For each arm $i\in[n]$, we add a terminal node $\phi_i$.  The arm transitions onto $\phi_i$ if, at a time step, we don't play it while it's on a non-root node.

Now we write the exponential-sized linear program for this problem.  Let $\cS_i'=\cS_i\cup\{\phi_i\}$ for all $i\in[n]$.  Let $\bcS'=\cS_1'\times\ldots\times\cS_n'\setminus\{\pi:\pi_i\notin\{\rho_i,\phi_i\},\pi_j\notin\{\rho_j,\phi_j\},i\neq j\}$, where we have excluded the joint nodes with two or more arms in the middle of being processed, since this is impossible without preemption.  For $\pi\in\bcS'$, let $I(\pi)=\{i:\pi_i\neq\phi_i\}$, the indices of arms that could be played from $\pi$.

The objective is
\begin{equation}
\label{E0'}
\max\sum_{\pi\in\bcS'}\sum_{i\in I(\pi)}\sum_{a\in A}r^a_{\pi_i}\sum_{t=1}^B z^a_{\pi,i,t}
\end{equation}
with very similar constraints on the $z$-variables:
\begin{subequations}
	\begin{align}
	\sum_{i\in I(\pi)}\sum_{a\in A}z^a_{\pi,i,t} & \le y_{\pi,t} && \pi\in\bcS',\ t\in[B] \label{E2a'} \\
	z^{\alpha}_{\pi,i,t} & =y_{\pi,t} && \pi\in\bcS',\ i:\pi_i\in\cB,\ t\in[B] \label{E2b'} \\
	z^a_{\pi,i,t} & \ge0 && \pi\in\bcS',\ i\in I(\pi),\ a\in A,\ t\in[B] \label{E3'}
	\end{align}
\end{subequations}
The only difference from (\ref{E2a})-(\ref{E3}) is that arms on terminal nodes cannot be played.  However, the state-updating constraints become more complicated, because now an arm can make a transition even while it is not being played, namely the transition to the terminal node.  Let $\bcA_i=\{\pi\in\bcS':\pi_i\notin\{\rho_i,\phi_i\}\}$, the joint nodes with arm $i$ in the middle of being processed.  We call arm $i$ the \emph{active} arm.  Let $\bcA=\bigcup_{i=1}^n\bcA_i$.  For $\pi\in\bcS'$, let $\bcP(\pi)$ denote the subset of $\bcS'$ that would transition to $\pi$ with no play: if $\pi\notin\bcA$, then $\bcP(\pi)=\{\pi\}\cup\big(\bigcup_{i\notin I(\pi)}\{\pi^u:u\in\cS_i\setminus\{\rho_i\}\}\big)$; if $\pi\in\bcA$, then $\bcP(\pi)=\emptyset$.  With this notation, we update the $y$-variables as follows:
\begin{subequations}
	\begin{align}
	y_{(\rho_1,\ldots,\rho_n),1} & =1 && \label{E4a'} \\
	y_{\pi,1} & =0 && \pi\in\bcS'\setminus\{(\rho_1,\ldots,\rho_n)\} \label{E4b'} \\
	y_{\pi,t} & =\sum_{\pi'\in\bcP(\pi)}\Big(y_{\pi',t-1}-\sum_{i\in I(\pi')}\sum_{a\in A}z^a_{\pi',i,t-1}\Big) && t>1,\ \pi\in\bcS'\setminus\bcA \label{E5a'} \\
	y_{\pi,t} & =\sum_{a:(\rho_i,a)\in\Par(\pi_i)}\Big(\sum_{\pi'\in\bcP(\pi^{\rho_i})}z^a_{\pi',i,t-1}\Big)\cdot p^a_{\rho_i,\pi_i} && t>1,\ i\in[n],\ \pi\in\bcA_i,\ \depth(\pi_i)=1 \label{E5b'} \\
	y_{\pi,t} & =\sum_{(u,a)\in\Par(\pi_i)}z^a_{\pi^u,i,t-1}\cdot p^a_{u,\pi_i} && t>1,\ i\in[n],\ \pi\in\bcA_i,\ \depth(\pi_i)>1 \label{E5c'}
	\end{align}
\end{subequations}

(\ref{E5a'}) updates $y_{\pi,t}$ for $\pi\notin\bcA$, ie.\ joint nodes with no active arms.  Such a joint node $\pi$ can only be arrived upon by making no play from a joint node in $\bcP(\pi)$.

(\ref{E5b'}), (\ref{E5c'}) update $y_{\pi,t}$ for $\pi\in\bcA$.  To get to joint node $\pi\in\bcA_i$, we must have played arm $i$ during the previous time step and transitioned to node $\pi_i$.  However, the restrictions on the previous joint node depend on whether $\depth(\pi_i)=1$.  If so, then arm $i$ was on $\rho_i$ at time step $t-1$, so it's possible to get to $\pi$ from any joint node in $\bcP(\pi^{\rho_i})$.  That is, in the previous joint node, there could have been an active arm that is not $i$.  This is reflected in (\ref{E5b'}).  On the other hand, if $\depth(\pi_i)>1$, then arm $i$ must have been the active arm at time step $t-1$, as described in (\ref{E5c'}).

Like before, these equations guarantee that at each time step, we are at exactly one joint node, ie.\ $\sum_{\pi\in\bcS'}y_{\pi,t}=1$.  Combined with (\ref{E2a'}), we obtain
\begin{align}
\sum_{\pi\in\bcS'}\sum_{i\in I(\pi)}\sum_{a\in A}z^a_{\pi,i,t} & \le1 && t\in[B] \label{E1'}
\end{align}

Let $\mathtt{(ExpLP')}$ denote the linear program defined by objective (\ref{E0'}) and constraints (\ref{E2a'})-(\ref{E3'}), (\ref{E4a'})-(\ref{E5c'}) which imply (\ref{E1'}).  This formally defines \emph{MAB superprocess with multi-period actions---no preemption}.

\subsection{Reductions from SK and MAB.}\label{subsection_with_diagrams}

Before we proceed, let's describe how this model generalizes the problems discussed in the introduction.  The generalization of the setting of Markovian Bandits (and its non-preempting analogue Irrevocable Bandits) is immediate: define the action set $A$ to only contain some default action $\alpha$, and don't have any bridge nodes.

For all variants of SK, the necessary reductions have already been described in Subsection~\ref{simple_reductions}, but the following examples should reiterate the power of our model:
\begin{itemize}
	\item Consider a job that takes time $5$ with probability $\frac{1}{3}$, and time $2$ with probability $\frac{2}{3}$ (and cannot be canceled once started).  If it finishes, the reward returned is $2$, independent of processing time.  This can be modeled by fig.~\ref{csk1} where $r_B=\frac{4}{3}$, $r_E=2$, and $\cB=\{B,C,D,E\}$.  Note that instead of placing reward $2$ on arc $(B,C')$, we have equivalently placed reward $\frac{4}{3}$ on node $B$.  A corollary of this reduction is that the following reward structure is equivalent to the original for the objective of maximizing expected reward: a guaranteed reward of $\frac{4}{3}$ after $2$ time steps, after which the job may run for another $3$ time steps to produce an additional $2$ reward.
	\item Consider the same job as the previous one, except the reward is $4$ if the processing time was $5$, while the reward is $1$ if the processing time was $2$ (the expected reward for finishing is still $2$).  All we have to change in the reduction is setting $r_B=\frac{2}{3}$ and $r_E=4$ instead.
	\item Consider either of the two jobs above, except cancellation is permitted (presumably on node $C$, after observing the transition from node $B$).  All we have to change in the reduction is setting $\cB=\emptyset$ instead.
	\item Consider the job from the second bullet that can be canceled, and furthermore, we find out after $1$ time step whether it will realize to the long, high-reward job or the short, low-reward job.  This can be modeled by fig.~\ref{csk2} where $r_{B'}=1$, $r_E=4$, and $\cB=\emptyset$.
\end{itemize}
\begin{figure}[t]
	\begin{center}
		\includegraphics[width=0.8\textwidth]{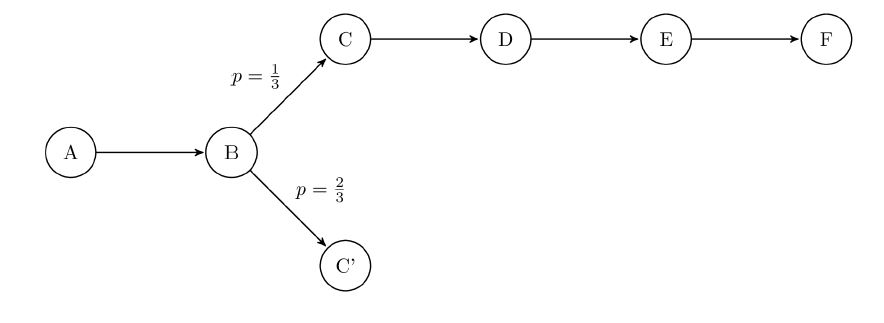}
		\caption{A Markov chain representing a SK job with correlated rewards.\label{csk1}}
	\end{center}
\end{figure}
\begin{figure}[t]
	\begin{center}
		\includegraphics[width=0.8\textwidth]{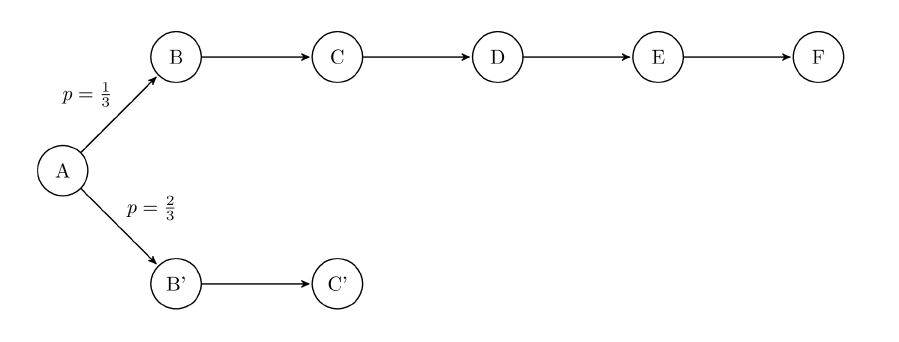}
		\caption{Another Markov chain representing a SK job with correlated rewards.\label{csk2}}
	\end{center}
\end{figure}
Whether preemption is allowed is determined at a global level, although this is irrelevant if no job can be canceled in the first place.  We would like to point out that preemption can be necessary for optimality even under the simplest setting of uncorrelated SK (Appendix~A), so disallowing preemption results in a distinct problem.

For Futuristic Bandits, suppose the exploration phase contains $T$ time steps.  Then we set $B=2T+1$ and add to each node a separate ``exploit'' action that returns reward $r_u$ in processing time $T+1$ (and there is no other way to obtain reward).  Clearly, we can explore for at most $T$ time steps if we are going to earn any reward at all\footnote{Note that it is never beneficial in a Martingale setting to not make full use of the $T$ exploration steps.}.  $B$ is chosen to be $2T+1$ so that it is impossible to collect exploitation rewards from more than one arm.  Budgeted Bandits can be modeled by combining the reductions for Stochastic Knapsack and Futuristic Bandits.

\subsection{Polynomial-sized LP relaxations.}\label{lp_relaxation_section}

We now write the polynomial-sized LP relaxations of our earlier problems.  We keep track of the probabilities of being on the nodes of each arm individually without considering their joint distribution.  Let $s_{u,t}$ be the probability arm $i$ is on node $u$ at the beginning of time $t$.  Let $x^a_{u,t}$ be the probability we play action $a$ on node $u$ at time $t$.

For both variants of the problem, we have the objective
\begin{equation}
\max\sum_{u\in\cS}\sum_{a\in A}r^a_u\sum_{t=1}^B x^a_{u,t} \label{P0}
\end{equation}
and constraints on how we can play each individual arm:
\begin{subequations}
	\begin{align}
	\sum_{a\in A}x^a_{u,t} & \le s_{u,t} && u\in\cS,\ t\in[B] \label{P2a} \\
	x^{\alpha}_{u,t} & =s_{u,t} && u\in\cB,\ t\in[B] \label{P2b} \\
	x^a_{u,t} & \ge0 && u\in\cS,\ a\in A,\ t\in[B] \label{P3}
	\end{align}
\end{subequations}
Furthermore, there is a single constraint 
\begin{align}
\sum_{u\in\cS}\sum_{a\in A}x^a_{u,t} & \le1 && t\in[B] \label{P1}
\end{align}
enforcing that the total probabilities of plays across all arms cannot exceed $1$ at any time step.

The state-updating constraints differ for the two variants of the problem.  If we allow preemption, then they are:
\begin{subequations}
	\begin{align}
	s_{\rho_i,1} & =1 && i\in[n] \label{P4a} \\
	s_{u,1} & =0 && u\in\cS\setminus\{\rho_1,\ldots,\rho_n\} \label{P4b} \\
	s_{u,t} & =s_{u,t-1}-\sum_{a\in A}x^a_{u,t-1}+\sum_{(v,a)\in\Par(u)}x^a_{v,t-1}\cdot p^a_{v,u} && t>1,\ u\in\cS \label{P5}
	\end{align}
\end{subequations}
If we disallow preemption, then an arm can only be on a non-root node if we played the same arm during the previous time step.  This is reflected in (\ref{P5a'})-(\ref{P5b'}):
\begin{subequations}
	\begin{align}
	s_{\rho_i,1} & =1 && i\in[n] \label{P4a'} \\
	s_{u,1} & =0 && u\in\cS\setminus\{\rho_1,\ldots,\rho_n\} \label{P4b'} \\
	s_{\rho_i,t} & =s_{\rho_i,t-1}-\sum_{a\in A}x^a_{\rho_i,t-1} && t>1,\ i\in[n] \label{P5a'} \\
	s_{u,t} & =\sum_{(v,a)\in\Par(u)}x^a_{v,t-1}\cdot p^a_{v,u} && t>1,\ u\in\cS\setminus\{\rho_1,\ldots,\rho_n\} \label{P5b'}
	\end{align}
\end{subequations}

Let $\mathtt{(PolyLP)}$ denote the linear program defined by objective (\ref{P0}) and constraints (\ref{P2a})-(\ref{P3}), (\ref{P1}), (\ref{P4a})-(\ref{P5}).  Similarly, let $\mathtt{(PolyLP')}$ denote the linear program defined by objective (\ref{P0}) and constraints (\ref{P2a})-(\ref{P3}), (\ref{P1}), (\ref{P4a'})-(\ref{P5b'}).  We still have to prove the polynomial-sized linear programs are indeed relaxations of the exponential-sized linear programs.  For any linear program LP, let $\OPT_{\mathtt{LP}}$ denote its optimal objective value.

\begin{lemma}
	\label{LPproj}
	Given a feasible solution $\{z^a_{\pi,i,t}\},\{y_{\pi,t}\}$ to $\mathtt{(ExpLP)}$, we can construct a solution to $\mathtt{(PolyLP)}$ with the same objective value by setting $x^a_{u,t}=\sum_{\pi\in\bcS:\pi_{i}=u}z^a_{\pi,i,t}$, $s_{u,t}=\sum_{\pi\in\bcS:\pi_{i}=u}y_{\pi,t}$ for all $i\in[n],u\in\cS_i,a\in A,t\in[B]$.  Thus the feasible region of $\mathtt{(PolyLP)}$ is a projection of that of $\mathtt{(ExpLP)}$ onto a subspace and $\OPT_{\mathtt{ExpLP}}\le\OPT_{\mathtt{PolyLP}}$.
\end{lemma}

\begin{lemma}
	\label{LPprimeproj}
	Given a feasible solution $\{z^a_{\pi,i,t}\},\{y_{\pi,t}\}$ to $\mathtt{(ExpLP')}$, we can construct a solution to $\mathtt{(PolyLP')}$ with the same objective value by setting $x^a_{u,t}=\sum_{\pi\in\bcS':\pi_{i}=u}z^a_{\pi,i,t}$, $s_{u,t}=\sum_{\pi\in\bcS':\pi_{i}=u}y_{\pi,t}$ for all $i\in[n],u\in\cS_i,a\in A,t\in[B]$.  Thus the feasible region of $\mathtt{(PolyLP')}$ is a projection of that of $\mathtt{(ExpLP')}$ onto a subspace and $\OPT_{\mathtt{ExpLP'}}\le \OPT_{\mathtt{PolyLP'}}$.
\end{lemma}

Recall that the feasible regions of the exponential-sized linear programs correspond exactly to the admissible policies.  These lemmas say that the performance of any adaptive policy can be upper bounded by the polynomial-sized relaxations.  Our lemmas are generalizations of similar statements from earlier works \cite[lem.~2.1---for example]{GKMR11}, but put into the context of an exponential-sized linear program.  The proofs are mostly technical and will be deferred to Appendix~B.

\subsection{Main results.} Now that we have established the preliminaries, we are ready to state our main results in the form of theorems.

\begin{theorem}
	\label{mainresult1}
	Given a feasible solution $\{x^a_{u,t}\},\{s_{u,t}\}$ to $\mathtt{(PolyLP')}$, there exists a solution to $\mathtt{(ExpLP')}$ with $\sum_{\pi:\pi_{i}=u}z^a_{\pi,i,t}=\frac{1}{2}x^a_{u,t}$, $\sum_{\pi:\pi_{i}=u}y_{\pi,t}=\frac{1}{2}s_{u,t}$ for all $i\in[n],u\in\cS_i,a\in A,t\in[B]$, obtaining reward $\frac{1}{2}\OPT_{\mathtt{PolyLP'}}$.  We can use sampling to turn this into a $(\frac{1}{2}-\varepsilon)$-approximation algorithm for MAB superprocess with multi-period actions---no preemption, with runtime polynomial in the input and $\frac{1}{\varepsilon}$.
\end{theorem}

We prove this theorem in Section~\ref{proof_section1}, and also show that it is tight, constructing an instance under the special case of correlated SK where it is impossible to obtain reward greater than $(\frac{1}{2}+\varepsilon)\OPT_{\mathtt{(PolyLP')}}$.

\begin{theorem}
	\label{mainresult2}
	There is a $\mathtt{(PolyLP)}$-relative $\frac{4}{27}$-approximation algorithm for MAB superprocess, where all processing times are $1$.
\end{theorem}

\begin{theorem}
	\label{mainresult3}
	There is a $\mathtt{(PolyLP)}$-relative $\frac{1}{12}$-approximation algorithm for MAB superprocess with multi-period actions.
\end{theorem}

We prove these theorems in Section~\ref{proof_section2}.

\section{Proof of Theorem~\ref{mainresult1}.}\label{proof_section1}

In this section we prove Theorem~\ref{mainresult1}.  To build intuition, we will first present the upper bound, showing a family of examples with $\frac{\OPT_{\mathtt{ExpLP'}}}{\OPT_{\mathtt{PolyLP'}}}$ approaching $\frac{1}{2}$.

\subsection{Construction for upper bound.}\label{upper_bound}

Let $N$ be a large integer.  We will describe our $n=2$ arms as stochastic jobs.  Job $1$ takes $N+1$ time with probability $1-\frac{1}{N}$, in which case it returns a reward of $1$.  It takes $1$ time with probability $\frac{1}{N}$, in which case it returns no reward.  Job $2$ deterministically takes $1$ time and returns a reward of $1$.  The budget is $B=N+1$ time steps.

Any actual policy can never get more than $1$ reward, since it cannot get a positive reward from both jobs.  To describe the solution to $(\mathtt{PolyLP'})$ earning more reward, we need to define some notation for the Markov chains representing the stochastic jobs.  A diagram for this reduction was shown in Subsection~\ref{subsection_with_diagrams}.

Let $\cS_1=\{S_0,S_1,\ldots,S_N,\phi_1\}$, with $\rho_1=S_0$.  There is only one action, and we will omit the action superscripts.  The only uncertainty is at $S_0$, with $p_{S_0,S_1}=1-\frac{1}{N},p_{S_0,\phi_1}=\frac{1}{N}$.  The remaining transitions are $p_{S_1,S_2}=\ldots=p_{S_{N-1},S_N}=p_{S_N,\phi_1}=1$, and self loop on the terminal node $p_{\phi_1,\phi_1}=1$.  The only reward is a reward of $1$ on node $S_n$.  Meanwhile, $\cS_2$ consists only of nodes $\{\rho_2,\phi_2\}$, with $p_{\rho_2,\phi_2}=p_{\phi_2,\phi_2}=1$, $r_{\rho_2}=1$.

It can be checked that $x_{S_0,1}=1,x_{S_1,2}=\ldots=x_{S_N,N+1}=1-\frac{1}{N},x_{\rho_2,2}=\ldots=x_{\rho_2,N+1}=\frac{1}{N}$ is a feasible solution for $\mathtt{(PolyLP')}$.  Its objective value is $2-\frac{1}{N}$, hence as we take $N\to\infty$, we get $\frac{\OPT_{\mathtt{ExpLP'}}}{\OPT_{\mathtt{PolyLP'}}}=\frac{1}{2}$.

Note that we can put all of $\cS_1\setminus\{\rho_1,\phi_1\}$ in $\cB$ if we want; it doesn't change the example whether job $1$ can be canceled once started.  It also doesn't matter whether we allow preemption---both $\frac{\OPT_{\mathtt{ExpLP}}}{\OPT_{\mathtt{PolyLP}}}$ and $\frac{\OPT_{\mathtt{ExpLP'}}}{\OPT_{\mathtt{PolyLP'}}}$ are $\frac{1}{2}+\varepsilon$ for this example.

Let's analyze what goes wrong when we attempt to replicate the optimal solution to the LP relaxation in an actual policy.  We start job $1$ at time $1$ with probability $x_{S_0,1}=1$.  If it does not terminate after $1$ time step, which occurs with probability $1-\frac{1}{N}$, then we play job $1$ through to the end, matching $x_{S_1,2}=\ldots=x_{S_N,N+1}=1-\frac{1}{N}$.  If it does, then we start job $2$ at time $2$.  This occurs with unconditional probability $x_{\rho_2,2}=\frac{1}{N}$, as planned.  However, in this case, we cannot start job $2$ again at time $3$ (since it has already been processed at time $2$), even though $x_{\rho_2,3}=\frac{1}{N}$ is telling us to do so.  The LP relaxation fails to consider that event ``job $1$ takes time $1$'' is \emph{directly correlated} with event ``job $2$ is started at time $2$'', so the positive values specified by $x_{\rho_2,3},\ldots,x_{\rho_2,N+1}$ are illegal plays.

Motivated by this example, we observe that if we only try to play $u$ at time $t$ with probability $\frac{x_u,t}{2}$, then we can obtain a solution to $\mathtt{(ExpLP')}$ (and hence a feasible policy) that is a scaled copy of the solution to $\mathtt{(PolyLP')}$.

\subsection{Technical specification of solution to (ExpLP').} Fix a solution $\{x^a_{u,t},s_{u,t}\}$ to $\mathtt{(PolyLP')}$.  Our objective in this subsection is to construct a solution $\{z^a_{\pi,i,t},y_{\pi,t}\}$ to $\mathtt{(ExpLP')}$ such that
\begin{align}
\sum_{\pi\in\bcS':\pi_i=u}z^a_{\pi,i,t} & =\frac{x^a_{u,t}}{2} && i\in[n],\ u\in\cS_i,\ a\in A \label{daggers}
\end{align}
obtaining half the objective value of $\mathtt{(PolyLP')}$.  We will prove feasibility in Subsection~\ref{feasibility}.

For convenience, define $x_{u,t}=\sum_{a\in A}x^a_{u,t}$ and $z_{\pi,i,t}=\sum_{a\in A}z^a_{\pi,i,t}$.  We will complete the specification of $\{z^a_{\pi,i,t},y_{\pi,t}\}$ over $B$ iterations $t=1,\ldots,B$.  On iteration $t$:
\begin{enumerate}
	\item Compute $y_{\pi,t}$ for all $\pi\in\bcS'$.
	\item Define $\ty_{\pi,t}=y_{\pi,t}$ if $\pi\notin\bcA$, and $\ty_{\pi,t}=y_{\pi,t}-\sum_{a\in A}z^a_{\pi,i,t}$ if $\pi\in\bcA_i$ for some $i\in[n]$ (if $\pi\in\bcA_i$, then $\{z^a_{\pi,i,t}:a\in A\}$ has already been set in a previous iteration).
	\item For all $i\in[n]$, define $f_{i,t}=\sum_{\pi\in\bcS':\pi_i=\rho_i}\ty_{\pi,t}$.
	\item For all $i\in[n]$, $\pi\in\bcS'$ such that $\pi_i=\rho_i$, and $a\in A$, set $z^a_{\pi,i,t}=\ty_{\pi,t}\cdot\frac{1}{2}\cdot\frac{x^a_{\rho_i,t}}{f_{i,t}}$.
	\item For all $i\in[n]$, and $\pi\in\bcS'$ such that $\pi_i=\rho_i$ and $\pi_j\in\{\rho_j,\phi_j\}$ for $j\neq i$, define $g_{\pi,i,t}=\sum_{\pi'\in\bcP(\pi)}z_{\pi',i,t}$.
	\item For all $i\in[n]$, $u\in\cS_i\setminus\{\rho_i\}$, $\pi\in\bcS'$ such that $\pi_i=u$, and $a\in A$, set $z^a_{\pi,i,t+\depth(u)}=g_{\pi^{\rho_i},i,t}\cdot\frac{x^a_{u,t+\depth(u)}}{x_{\rho_i,t}}$.
\end{enumerate}

The motivation for Step~2 is that we want $\ty_{\pi,t}$ to represent the probability we are at joint node $\pi$ and looking to start a new arm at time $t$, abandoning the arm in progress if there is any.  In Step~3, $f_{i,t}$ is the total probability we would be ready to start playing arm $i$ at time $t$.  The normalization in Step~4 ensures that each arm is started with the correct probability at time $t$.  In Step~5, $g_{\pi,i,t}$ is the probability arm $i$ is started at time $t$, and other arms are on nodes $\{\pi_j:j\neq i\}$ while arm $i$ executes (another arm $j$ could have made a transition to $\phi_j$ during the first time step $t$).  Step~6 specifies how to continue playing arm $i$ in subsequent time steps if it is started at time $t$.  Note that $g_{\pi^{\rho_i},i,t}$ is guaranteed to be defined in this case, since $\pi_i\notin\{\rho_i,\phi_i\}$ and $\pi\in\bcS'$ implies $\pi_j\in\{\rho_j,\phi_j\}$ for all $j\neq i$.

This completes the specification of the solution to $\mathtt{(ExpLP')}$.  Every $y_{\pi,t}$ is set in Step~1, and every $z^a_{\pi,i,t}$ is set in either Step~4 or Step~6.

Using the definition of $f_{i,t}$, Step~4 guarantees that for $i\in[n]$, $a\in A$,
\begin{eqnarray*}
	\sum_{\pi\in\bcS':\pi_i=\rho_i}z^a_{\pi,i,t} & = & \sum_{\pi\in\bcS':\pi_i=\rho_i}\ty_{\pi,t}\cdot\frac{1}{2}\cdot\frac{x^a_{\rho_i,t}}{f_{i,t}} \\
	& = & \frac{x^a_{\rho_i,t}}{2}
\end{eqnarray*}

Meanwhile, Step~6 guarantees that for $i\in[n]$, $u\in\cS_i\setminus\{\rho_i\}$, $a\in A$,
\begin{eqnarray*}
	\sum_{\pi\in\bcS':\pi_i=u}z^a_{\pi,i,t+\depth(u)} & = & \sum_{\pi\in\bcS':\pi_i=u}g_{\pi,i,t}\cdot\frac{x^a_{u,t+\depth(u)}}{x_{\rho_i,t}} \\
	& = & \frac{x_{\rho_i,t}}{2}\cdot\frac{x^a_{u,t+\depth(u)}}{x_{\rho_i,t}} \\
	& = & \frac{x^a_{u,t+\depth(u)}}{2}
\end{eqnarray*}
We explain the second equality.  Since $u\neq\rho_i$ implies arm $i$ is the active arm in all of $\{\pi\in\bcS':\pi_i=u\}$, this set is equal to $\{\rho_1,\phi_1\}\times\cdots\times u\times\cdots\times\{\rho_n,\phi_n\}$.  Summing $g_{\pi,i,t}$ over all the possibilities for $\{\pi_j:j\neq i\}$ yields the total probability arm $i$ is started at time $t$.  This is equal to $\sum_{\pi\in\bcS':\pi_i=\rho_i}\sum_{a\in A}z^a_{\pi,i,t}$, which by the first calculation is equal to $\frac{x_{\rho_i,t}}{2}$.

The proof of (\ref{daggers}) is now complete.

\subsection{Proof of feasibility.}\label{feasibility}

We will inductively prove feasibility over iterations $t=1,\ldots,B$.  Suppose all of the variables $\{z^a_{\pi,i,t'},y_{\pi,t'}\}$ with $t'<t$ have already been set in a way that satisfies constraints (\ref{E2a'})-(\ref{E3'}), (\ref{E4a'})-(\ref{E5c'}).  Some of the variables $z^a_{\pi,i,t'}$ with $t'\ge t$ may have also been set in Step~6 of earlier iterations; if so, suppose they have already been proven to satisfy (\ref{E3'}).

On iteration $t$, we first compute in Step~1 $y_{\pi,t}$ for all $\pi\in\bcS'$; these are guaranteed to satisfy (\ref{E4a'})-(\ref{E5c'}) by definition.  To complete the induction, we need to show that (\ref{E2a'})-(\ref{E3'}) hold after setting the $z$-variables in Step~4, and furthermore, (\ref{E3'}) holds for any $z^a_{\pi,i,t'}$ (with $t'>t$) we set in Step~6.

We first prove the following lemma:
\begin{lemma}
	\label{tilde_y_lemma}
	Suppose $\pi\in\bcA_i$ for some $i\in[n]$.  Let $u=\pi_i$ (which is neither $\rho_i$ nor $\phi_i$).  Then $\sum_{a\in A}z^a_{\pi,i,t}\le y_{\pi,t}$, and furthermore if $u\in\cB$, then $z^\alpha_{\pi,i,t}=y_{\pi,t}$.
\end{lemma}
\proof{Proof.}
First suppose $\depth(u)=1$.  (\ref{E5b'}) says $y_{\pi,t} = \sum_{a:(\rho_i,a)\in\Par(u)}(\sum_{\pi'\in\bcP(\pi^{\rho_i})}z^a_{\pi',i,t-1})\cdot p^a_{\rho_i,u}$.  Every $\pi'$ in the sum has $\pi'_i=\rho_i$, so $z^a_{\pi',i,t-1}$ was set in Step~4 of iteration $t-1$ to $\ty_{\pi',t-1}\cdot\frac{1}{2}\cdot\frac{x^a_{\rho_i,t-1}}{f_{i,t-1}}$.  Substituting into (\ref{E5b'}), we get
\begin{eqnarray*}
	y_{\pi,t} & = & \sum_{a:(\rho_i,a)\in\Par(u)}\Big(\sum_{\pi'\in\bcP(\pi^{\rho_i})}\ty_{\pi',t-1}\cdot\frac{1}{2}\cdot\frac{x^a_{\rho_i,t-1}}{f_{i,t-1}}\Big)\cdot p^a_{\rho_i,u} \\
	& = & \Big(\sum_{\pi'\in\bcP(\pi^{\rho_i})}\ty_{\pi',t-1}\Big)\cdot\frac{1}{2f_{i,t-1}}\cdot\sum_{a:(\rho_i,a)\in\Par(u)}x^a_{\rho_i,t-1}\cdot p^a_{\rho_i,u}
\end{eqnarray*}
Meanwhile, for all $a\in A$, $z^a_{\pi,i,t}$ was set in Step~6 of iteration $t-1$ to $g_{\pi^{\rho_i},i,t-1}\cdot\frac{x^a_{u,t}}{x_{\rho_i,t-1}}$.  Hence
\begin{eqnarray*}
	z^a_{\pi,i,t} & = & g_{\pi^{\rho_i},i,t-1}\cdot\frac{x^a_{u,t}}{x_{\rho_i,t-1}} \\
	& = & \sum_{\pi'\in\bcP(\pi^{\rho_i})}z_{\pi',i,t-1}\cdot\frac{x^a_{u,t}}{x_{\rho_i,t-1}} \\
	& = & \sum_{\pi'\in\bcP(\pi^{\rho_i})}\Big(\sum_{b\in A}\ty_{\pi',t-1}\cdot\frac{1}{2}\cdot\frac{x^b_{\rho_i,t-1}}{f_{i,t-1}}\Big)\cdot\frac{x^a_{u,t}}{x_{\rho_i,t-1}} \\
	& = & \Big(\sum_{\pi'\in\bcP(\pi^{\rho_i})}\ty_{\pi',t-1}\Big)\cdot\frac{1}{2f_{i,t-1}}\cdot x^a_{u,t}
\end{eqnarray*}
where the second equality is by the definition of $g_{\pi^{\rho_i},i,t-1}$, and the third equality uses the fact that $z^b_{\pi',i,t-1}$ was set in Step~4 of iteration $t-1$.  To prove $\sum_{a\in A}z^a_{\pi,i,t}\le y_{\pi,t}$, it suffices to show $\sum_{a\in A}x^a_{u,t}\le\sum_{a:(\rho_i,a)\in\Par(u)}x^a_{\rho_i,t-1}\cdot p^a_{\rho_i,u}$.  This follows immediately from combining constraints (\ref{P2a}) and (\ref{P5b'}) of $\mathtt{(PolyLP')}$.  Furthermore, if $u\in\cB$, then we can use (\ref{P2b}) to get $z^{\alpha}_{\pi,i,t}=y_{\pi,t}$.

Now suppose $\depth(u)>1$.  (\ref{E5c'}) says $y_{\pi,t}=\sum_{(v,a)\in\Par(u)}z^a_{\pi^v,i,t-1}\cdot p^a_{v,u}$.  Since $v\neq\rho_i$, $z^a_{\pi^v,i,t-1}$ was set in Step~6 of iteration $t':=t-\depth(u)$ to $g_{\pi^{\rho_i},i,t'}\cdot\frac{x^a_{v,t-1}}{x_{\rho_i,t'}}$.  Substituting into (\ref{E5c'}), we get
\begin{eqnarray*}
	y_{\pi,t} & = & \sum_{(v,a)\in\Par(u)}g_{\pi^{\rho_i},i,t'}\cdot\frac{x^a_{v,t-1}}{x_{\rho_i,t'}}\cdot p^a_{v,u} \\
	& = & \frac{g_{\pi^{\rho_i},i,t'}}{x_{\rho_i,t'}}\sum_{(v,a)\in\Par(u)}x^a_{v,t-1}\cdot p^a_{v,u}
\end{eqnarray*}
Meanwhile, for all $a\in A$, $z^a_{\pi,i,t}$ was set in Step~6 of iteration $t'$ to $g_{\pi^{\rho_i},i,t'}\cdot\frac{x^a_{u,t}}{x_{\rho_i,t'}}$.  To prove $\sum_{a\in A}z^a_{\pi,i,t}\le y_{\pi,t}$, it suffices to show $\sum_{a\in A}x^a_{u,t}\le\sum_{(v,a)\in\Par(u)}x^a_{v,t-1}\cdot p^a_{v,u}$.  This is again obtained from (\ref{P2a}) and (\ref{P5b'}), and if $u\in\cB$, then we can use (\ref{P2b}) to get $z^{\alpha}_{\pi,i,t}=y_{\pi,t}$.
\Halmos\endproof

By the lemma, $\ty_{\pi,t}\ge0$ if $\pi\in\bcA_i$ for some $i\in[n]$.  On the other hand, $\ty_{\pi,t}\ge0$ is immediate from definition if $\pi\notin\bcA$.  Therefore, $\ty_{\pi,t}\ge0$ for all $\pi\in\bcS'$, and (\ref{E3'}) is satisfied by all the $z$-variables set in Step~4 or Step~6.  Furthermore, the lemma guarantees (\ref{E2b'}) for the $z^a_{\pi,i,t}$ with $\pi_i\in\cB$ set in previous iterations.

It remains to prove (\ref{E2a'}).  If $\pi\in\bcA_i$, then the LHS of $(\ref{E2a'})$ is
\begin{eqnarray*}
	\sum_{j\in I(\pi)}z_{\pi,j,t} & = & z_{\pi,i,t}+\sum_{j\in I(\pi)\setminus\{i\}}\ty_{\pi,t}\cdot\frac{1}{2}\cdot\frac{x_{\rho_j,t}}{f_{j,t}} \\
	& = & z_{\pi,i,t}+(y_{\pi,t}-z_{\pi,i,t})\cdot\sum_{j\in I(\pi)\setminus\{i\}}\frac{1}{2}\cdot\frac{x_{\rho_j,t}}{f_{j,t}}
\end{eqnarray*}
For the first equality, note that $z_{\pi,j,t}$ for $j\neq i$ is set in Step~4 of the current iteration, but $z_{\pi,i,t}$ has already been set in an earlier iteration.  The second equality is immediate from the definition of $\ty_{\pi,t}$.  Note that $y_{\pi,t}-z_{\pi,i,t}\ge0$, by Lemma~\ref{tilde_y_lemma}.  If we knew $\sum_{j\in I(\pi)\setminus\{i\}}\frac{1}{2}\cdot\frac{x_{\rho_j,t}}{f_{j,t}}\le1$, then we would have $\sum_{j\in I(\pi)}z_{\pi,j,t}\le z_{\pi,i,t}+(y_{\pi,t}-z_{\pi,i,t})(1)=y_{\pi,t}$, which is (\ref{E2a'}).

On the other hand, if $\pi\notin\bcA$, then the LHS of $(\ref{E2a'})$ is $y_{\pi,t}\cdot\sum_{j\in I(\pi)}\frac{1}{2}\cdot\frac{x_{\rho_j,t}}{f_{j,t}}$, where in this case all of the $z_{\pi,j,t}$ are set in Step~4 of the current iteration.  Similarly, if we knew $\sum_{j\in I(\pi)}\frac{1}{2}\cdot\frac{x_{\rho_j,t}}{f_{j,t}}\le1$, then we would have (\ref{E2a'}).

To complete the proof of feasibility, it suffices to show $\sum_{j=1}^n\frac{1}{2}\cdot\frac{x_{\rho_j,t}}{f_{j,t}}\le 1$ (note that $f_{j,t}$ is always non-negative, by its definition and Lemma~\ref{tilde_y_lemma}).  This is implied by the following lemma, which proves a simpler statement:
\begin{lemma}
	\label{enough_space_exact}
	$f_{i,t}\ge\sum_{j=1}^n\frac{x_{\rho_j,t}}{2}$ for all $i\in[n]$.
\end{lemma}

\proof{Proof.}
Fix some $i\in[n]$.  By the definitions in Step~2 and Step~3,
\begin{equation*}
f_{i,t}=\sum_{\pi\in\bcS':\pi_i=\rho_i}y_{\pi,t}-\sum_{j\neq i}\sum_{\pi\in\bcA_j:\pi_i=\rho_i}z_{\pi,j,t}
\end{equation*}

Let's start by bounding $\sum_{\pi\in\bcS':\pi_i=\rho_i}y_{\pi,t}$, the total probability arm $i$ is still on $\rho_i$ at the start of time $t$.  This is equal to $1-\sum_{t'<t}\sum_{\pi\in\bcS':\pi_i=\rho_i}z_{\pi,i,t'}$, where we subtract from $1$ the total probability arm $i$ was initiated before time $t$.  By (\ref{daggers}), $\sum_{\pi\in\bcS':\pi_i=\rho_i}z_{\pi,i,t'}=\frac{x_{\rho_i,t'}}{2}$ for all $t'<t$.  Furthermore,
\begin{equation}
\label{root_unity}
\sum_{t'<t}x_{\rho_i,t'}\le1
\end{equation}
from iteratively applying (\ref{P5a'}) to (\ref{P4a'}), and combining with (\ref{P2a})\footnote{Intuitively, we're arguing that solution to the LP relaxation still satisfies the total probability arm $i$ being played from its root node not exceeding unity.}.  Therefore, $\sum_{\pi\in\bcS':\pi_i=\rho_i}y_{\pi,t}\ge\frac{1}{2}$.

Now we bound the remaining term in the equation for $f_{i,t}$:
\begin{eqnarray*}
	\sum_{j\neq i}\sum_{\pi\in\bcA_j:\pi_i=\rho_i}z_{\pi,j,t} & = & \sum_{j\neq i}\sum_{v\in\cS_j\setminus\{\rho_j\}}\sum_{\pi:\pi_i=\rho_i,\pi_j=v}z_{\pi,j,t} \\
	& \le & \sum_{j\neq i}\sum_{v\in\cS_j\setminus\{\rho_j\}}\sum_{\pi:\pi_j=v}z_{\pi,j,t} \\
	& = & \frac{1}{2}\sum_{j\neq i}\sum_{v\in\cS_j\setminus\{\rho_j\}}x_{v,t} \\
	& \le & \frac{1}{2}\sum_{j=1}^n\sum_{v\in\cS_j\setminus\{\rho_j\}}x_{v,t} \\
	& \le & \frac{1}{2}\big(1-\sum_{j=1}^n x_{\rho_j,t}\big) \\
\end{eqnarray*}
The first inequality uses the non-negativity of $z_{\pi,j,t}$ in the inductively proven (\ref{E3'}), the second equality uses (\ref{daggers}), the second inequality uses the non-negativity of $x_{v,t}$ in (\ref{P3}), and the third inequality uses (\ref{P1}).

Combining the two terms, we get $f_{i,t}\ge\sum_{j=1}^n\frac{x_{\rho_j,t}}{2}$, as desired.
\Halmos\endproof

\subsection{Approximation algorithm via sampling.} Let's turn this exponential-sized solution $\{z^a_{\pi,i,t},y_{\pi,t}\}$ of $\mathtt{(ExpLP')}$ into a polynomial-time policy.  Now we will assume the $\{x^a_{u,t},s_{u,t}\}$ we are trying to imitate is an \emph{optimal} solution of $\mathtt{(PolyLP')}$.  Consider the following algorithm, which takes in as parameters a terminal time step $t\in[B]$, and probabilities $\lambda_{i,t'}$ for each $i\in[n],t'\le t$ (which for now should be considered to be $f_{i,t'}$ to aid in the comprehension of the algorithm):

\

\noindent $\subr$($t$, $\{\lambda_{i,t'}:i\in[n],t'\le t\}$)
\begin{itemize}
	\item Initialize $t'=1$, $\current=0$.
	\item While $t'\le t$:
	\begin{enumerate}
		\item If $\current=0$, then
		\begin{enumerate}
			\item For each arm $i$ that is on $\rho_i$, set $\current=i$ with probability $\frac{1}{2}\cdot\frac{x_{\rho_i,t'}}{\lambda_{i,t'}}$; if the sum of these probabilities exceeds $1$ (ie.\ this step is inadmissible), then terminate with no reward.
			\item If $\current$ was set in this way, leave $t'$ unchanged and enter the next if block.  Otherwise, leave $\current$ at $0$ but increment $t'$ by $1$.
		\end{enumerate}
		\item If $\current\neq 0$, then
		\begin{enumerate}
			\item Let $u$ denote the node arm $\current$ is on.  For each $a\in A$, play action $a$ on arm $\current$ with probability $\frac{x^a_{u,t'}}{x_{u,t'}}$.
			\item Suppose we transition onto node $v$ as a result of this play.  With probability $\frac{x_{v,t'+1}}{s_{v,t'+1}}$, leave $\current$ unchanged.  Otherwise, set $\current=0$.
			\item Increment $t'$ by $1$.
		\end{enumerate}
	\end{enumerate}
\end{itemize}

\

Define the following events and probabilities, which depend on the input passed into $\subr$:
\begin{itemize}
	\item For all $i\in[n],t'\le(t+1)$, let $\cA_{i,t'}$ be the event that at the beginning of time $t'$, $\current=0$ and arm $i$ is on $\rho_i$.  Let $\Free(i,t')=\Pr[\cA_{i,t'}]$.
	\item For all $i\in[n],t'\le t$, let $\started(i,t')$ be the probability that we play arm $i$ from $\rho_i$ at time $t'$.
	\item For all $u\in\cS,a\in A,t'\le t$, let $\played(u,a,t')$ be the probability that we play action $a$ on node $u$ at time $t'$.
\end{itemize}

It is easy to see that $\subr$ is an algorithmic specification of feasible solution $\{z^a_{\pi,i,t},y_{\pi,t}\}$ if we run it on input ($B$, $\{f_{i,t}:i\in[n],t\in[B]\}$).  Indeed, we would iteratively have for $t=1,\ldots,B$:
\begin{itemize}
	\item $\Free(i,t)=f_{i,t}$ for all $i\in[n]$
	\item $\started(i,t)=\Free(i,t)\cdot\frac{1}{2}\cdot\frac{x_{\rho_i,t}}{f_{i,t}}=\frac{x_{\rho_i,t}}{2}$ for all $i\in[n]$
	\item $\played(u,a,t)=\started(i,t-\depth(u))\cdot\frac{x^a_{u,t}}{x_{\rho_i,t-\depth(u)}}=\frac{x^a_{u,t}}{2}$ for all $u\in\cS,a\in A$
\end{itemize}
The final statement can be seen inductively:
\begin{eqnarray}
\played(u,a,t) & = & \sum_{(v,b)\in\Par(u)}\played(v,b,t-1)\cdot p^b_{v,u}\cdot\frac{x^a_{u,t}}{s_{u,t}} \nonumber \\
& = & \sum_{(v,b)\in\Par(u)}\big(\started(i,t-1-\depth(v))\cdot\frac{x^b_{v,t-1}}{x_{\rho_i,t-1-\depth(v)}}\big)\cdot p^b_{v,u}\cdot\frac{x^a_{u,t}}{s_{u,t}} \nonumber \\
& = & \started(i,t-\depth(u))\cdot\frac{x^a_{u,t}}{x_{\rho_i,t-\depth(u)}} \label{nonroot_nodes}
\end{eqnarray}
where the first equality is by Steps~2a-b, the second equality is by the induction hypothesis, and the final equality is by (\ref{P5b'}).

Therefore, we would have a $\frac{1}{2}$-approximation if we knew $\{f_{i,t}:i\in[n],t\in[B]\}$, but unfortunately computing $f_{i,t}$ requires summing exponentially many terms.  We can try to approximate it by sampling, but we can't even generate a sample from the binary distribution with probability $f_{i,t}$ since that requires knowing the exact values of $f_{i,t'}$ for $t'<t$.  So we give up trying to approximate $f_{i,t}$, and instead iteratively approximate the values of $\Free(i,t)$ when $\subr$ is ran on previously approximated $\Free(i,t)$ values.

Fix some small $\varepsilon,\delta>0$ that will be determined later.  Let $\med=\frac{3\ln(2\delta^{-1})}{\varepsilon^2}$.  Change $\subr$ so that the probabilities in Step~1a are multiplied by $(1-\varepsilon)^2$ (and change the definitions of $\cA_{i,t'},\Free,\started,\played$ accordingly).

\

\noindent $\mathsf{Sampling\ Algorithm}$
\begin{itemize}
	\item Initialize $\Freeemp(i,1)=1$ for all $i\in[n]$.
	\item For $t=2,\ldots,B$:
	\begin{enumerate}
		\item Run $\subr$($t-1$, $\{\Freeemp(i,t'):i\in[n],t'<t\}$) a total of $M=\frac{8|\cS|B}{\varepsilon}\cdot\med$ times.  For all $i\in[n]$, let $C_{i,t}$ count the number of times event $\cA_{i,t}$ occurred.
		\item For each $i\in[n]$, if $C_{i,t}>\med$, set $\Freeemp(i,t)=\frac{C_{i,t}}{M}$; otherwise set $\Freeemp(i,t)=\sum_{j=1}^n\frac{x_{\rho_j,t}}{2}$.
	\end{enumerate}
\end{itemize}

\

Consider iteration $t$ of $\mathsf{Sampling\ Algorithm}$.  $\{\Freeemp(i,t'):i\in[n],t'<t\}$ have already been finalized, and we are sampling event $\cA_{i,t}$ when (the $\varepsilon$-modified) $\subr$ is ran on those finalized approximations to record values for $\{\Freeemp(i,t):i\in[n]\}$.  For all $i\in[n]$, if $C_{i,t}>\med$, then the probability of $\frac{C_{i,t}}{M}$ lying in $\big((1-\varepsilon)\cdot\Free(i,t),(1+\varepsilon)\cdot\Free(i,t)\big)$ is at least\footnote{This is because $\Freeemp(i,t)$ is an average over $M\ge\frac{8}{\varepsilon}\cdot\med$ runs, which is enough samples to guarantee this probability; see Motwani and Raghavan \cite{MR10}.} $1-\delta$.  As far as when we have $C_{i,t}>\med$, note that if $\Free(i,t)>\frac{\varepsilon}{4|\cS|B}$, then $\mathbb{E}[C_{i,t}]>2\med$, so the Chernoff bound says $\Pr[C_{i,t}\le\med]=O(\delta^{\frac{1}{\varepsilon^2}})=O(\delta)$.  We have discussed two $O(\delta)$ probability events in this paragraph of sampling/Chernoff yielding an unlikely and undesired result; call these events \emph{failures}.

By the union bound, the probability of having any failure over iterations $t=2,\ldots,B$ is at most $2(B-1)n(\delta+O(\delta))=O(Bn\delta)$.  Assuming no failures, we will inductively prove
\begin{equation}
\label{sampling_bound}
\frac{(1-\varepsilon)^2}{1+\varepsilon}\cdot\frac{x_{\rho_i,t}}{2}\le\started(i,t)\le\max\big\{(1-\varepsilon)\cdot\frac{x_{\rho_i,t}}{2},\frac{\varepsilon}{4|\cS|B}\big\}
\end{equation}
for all $i\in[n]$.  This is clear when $t=1$ since $\started(i,1)=\frac{x_{\rho_i,1}}{2}$ exactly for all $i\in[n]$.

Now suppose $t\ge2$.  We will first prove a lemma on the true probabilities $\Free(i,t)$, which is the ``approximate'' version of Lemma~\ref{enough_space_exact}:
\begin{lemma}
	\label{enough_space}
	Suppose $\subr$ is ran on input $(t-1,\{\Freeemp(i,t'):i\in[n],t'<t\})$ and there were no failures while obtaining the sample average approximations $\Freeemp(i,t')$.  Then for all $i\in[n]$, $\Free(i,t)\ge\frac{1}{2}\sum_{j=1}^nx_{\rho_j,t}$.
\end{lemma}

\proof{Proof.}
We know that event $\cA_{i,t}$ will occur if at time $t$, arm $i$ has not yet been started, and no other arm is active.  By the union bound, $1-\Free(i,t)\le\sum_{t'<t}\started(i,t')+\sum_{j=1}^n\sum_{u\in\cS_j\setminus\{\rho_j\}}\sum_{a\in A}\played(u,a,t)$.  Assuming (\ref{sampling_bound}) holds, we can bound
\begin{eqnarray*}
	\sum_{t'<t}\started(i,t') & \le & \sum_{t'<t}\big((1-\varepsilon)\cdot\frac{x_{\rho_i,t'}}{2}+\frac{\varepsilon}{4|\cS|B}\big) \\
	& \le & \frac{1-\varepsilon}{2}\sum_{t'<t}x_{\rho_i,t'}+\frac{\varepsilon}{4|\cS|} \\
	& \le & \frac{1-\varepsilon}{2}+\frac{\varepsilon}{4}
\end{eqnarray*}
where the final inequality uses (\ref{root_unity}).  Similarly, assuming (\ref{sampling_bound}) holds, we can bound
\begin{eqnarray*}
	\sum_{j=1}^n\sum_{u\in\cS_j\setminus\{\rho_j\}}\sum_{a\in A}\played(u,a,t) & = & \sum_{j=1}^n\sum_{u\in\cS_j\setminus\{\rho_j\}}\started(i,t-\depth(u))\cdot\frac{x_{u,t}}{x_{\rho_i,t-\depth(u)}} \\
	& \le & \sum_{j=1}^n\sum_{u\in\cS_j\setminus\{\rho_j\}}\big((1-\varepsilon)\cdot\frac{x_{\rho_i,t-\depth(u)}}{2}+\frac{\varepsilon}{4|\cS|B}\big)\cdot\frac{x_{u,t}}{x_{\rho_i,t-\depth(u)}} \\
	& \le & \frac{1-\varepsilon}{2}\sum_{j=1}^n\sum_{u\in\cS_j\setminus\{\rho_j\}}x_{u,t}+\frac{\varepsilon}{4B} \\
	& \le & \frac{1}{2}\Big(1-\sum_{j=1}^nx_{\rho_j,t}\Big)+\frac{\varepsilon}{4} \\
\end{eqnarray*}
where the equality uses (\ref{nonroot_nodes}), the second inequality uses the fact that $x_{u,t}\le x_{\rho_i,t-\depth(u)}$, and the final inequality uses (\ref{P1}).  Combining these bounds completes the proof of the lemma.
\Halmos\endproof

By the description in Step~1a of $\subr$, for all $i\in[n]$, we have
\begin{equation}
\label{desc_of_alg}
\started(i,t)=\Free(i,t)\cdot\frac{1}{2}\cdot\frac{x_{\rho_i,t}}{\Freeemp(i,t)}\cdot(1-\varepsilon)^2
\end{equation}
If $C_{i,t}>\med$, then $\Freeemp(i,t)$ will be set to $\frac{C_{i,t}}{M}$, and furthermore no failures implies $(1-\varepsilon)\cdot\Free(i,t)\le\frac{C_{i,t}}{M}\le(1+\varepsilon)\cdot\Free(i,t)$.  Substituting into (\ref{desc_of_alg}), we get $\frac{(1-\varepsilon)^2}{1+\varepsilon}\cdot\frac{x_{\rho_i,t}}{2}\le\started(i,t)\le(1-\varepsilon)\cdot\frac{x_{\rho_i,t}}{2}$ which implies (\ref{sampling_bound}).  On the other hand, if $C_{i,t}\le\med$, then $\Freeemp(i,t)$ will be set to $\sum_{j=1}^n\frac{x_{\rho_j,t}}{2}$, and assuming no failures it must have been the case that $\Free(i,t)\le\frac{\varepsilon}{4|\cS|B}$.  Substituting into (\ref{desc_of_alg}), we get $\started(i,t)\le\frac{\varepsilon}{4|\cS|B}\cdot\frac{1}{2}\cdot\frac{x_{\rho_i,t}}{\sum_{j=1}^nx_{\rho_j,t}}\cdot(1-\varepsilon)^2\le\frac{\varepsilon}{4|\cS|B}$ which implies the upper bound in (\ref{sampling_bound}).  For the lower bound, Lemma~\ref{enough_space} says $\Free(i,t)\ge\Freeemp(i,t)$, so $\started(i,t)\ge(1-\varepsilon)^2\cdot\frac{x_{\rho_i,t}}{2}\ge\frac{(1-\varepsilon)^2}{1+\varepsilon}\cdot\frac{x_{\rho_i,t}}{2}$.

This completes the induction for (\ref{sampling_bound}).  The final thing to check is that with these new parameters $\{\Freeemp(i,t):i\in[n]\}$, the sum of the probabilities in Step~1a of $\subr$ does not exceed $1$.  $\Freeemp(i,t)$ will either get set to $\sum_{j=1}^n\frac{x_{\rho_j,t}}{2}$, or be at least $(1-\varepsilon)\cdot\Free(i,t)$, which is at least $(1-\varepsilon)\cdot\sum_{j=1}^n\frac{x_{\rho_j,t}}{2}$ by Lemma~\ref{enough_space}.  In either case, $\Freeemp(i,t)\ge(1-\varepsilon)\cdot\sum_{j=1}^n\frac{x_{\rho_j,t}}{2}$ for all $i\in[n]$, so the desired sum in Step~1a is at most $\frac{1}{1-\varepsilon}\cdot(1-\varepsilon)^2\le1$.

We have an algorithm that fails with probability $O(Bn\delta)$, and when it doesn't fail, $\started(i,t)\ge\frac{(1-\varepsilon)^2}{1+\varepsilon}\cdot\frac{x_{\rho_i,t}}{2}$ for all $i\in[n],t\in[B]$, which in conjunction with (\ref{nonroot_nodes}) shows that we obtain expected reward at least $\frac{(1-\varepsilon)^2}{1+\varepsilon}\cdot\frac{1}{2}\cdot\OPT_{\mathtt{PolyLP'}}$.  Recall from Lemma~\ref{LPprimeproj} that $\OPT_{\mathtt{PolyLP'}}\ge\OPT_{\mathtt{ExpLP'}}$.  Treating a failed run as a run with $0$ reward, we can set $\delta=\Theta(\frac{\varepsilon}{Bn})$ to get a $(\frac{1}{2}-\varepsilon)$-approximation.  Finally, note that the runtime of this approximation algorithm is polynomial in the input, $\frac{1}{\varepsilon}$, and $\ln(\frac{1}{\delta})$, completing the proof of Theorem~\ref{mainresult1}.

\section{Proof of Theorem~\ref{mainresult2}.}\label{proof_section2}

In this section we prove Theorem~\ref{mainresult2}, and also show how to modify the proof to prove Theorem~\ref{mainresult3}.

\subsection{Description of algorithm.}\label{descr_priority_alg}

The algorithm maintains a priority index\footnote{For mathematical convenience, lower priority indices will mean higher priorities.} for each arm, telling us when it will try to play that arm again.  Formally, if an arm is on node $u$, we say the arm is in \emph{status} $(u,a,t)$ to represent that the algorithm will next try to play action $a\in A$ on the arm at time $t\in[B]$.  We allow $t=\infty$ to indicate that the algorithm will never try to play the arm again; in this case we omit the action parameter.

Fix an optimal solution $\{x^a_{u,t},s_{u,t}\}$ to $\mathtt{(PolyLP)}$.  We initialize each arm $i$ to status $(\rho_i,a,t)$ with probability $\frac{x^a_{\rho_i,t}}{C}$, for all $a\in A$ and $t\in[B]$, where $C>0$ is some constant which we optimize later.  With probability $1-\sum_{a\in A}\sum_{t=1}^B\frac{x^a_{\rho_i,t}}{C}$, the arm is initialized to status $(\rho_i,\infty)$ and never touched; note that this probability is at least $1-\frac{1}{C}$.

The statuses also evolve according to the solution of the LP relaxation.  If we play an arm and it transitions to node $u$, we need to decide what status $(u,a,t)$ to put that arm in.  The evolution of statuses is independent of other arms.  For all $i\in[n]$, $u\in\cS_i\setminus\{\rho_i\}$, $a\in A$, $t\in[B]$, $(v,b)\in\Par(u)$, and $t'<t$, define $q_{v,b,t',u,a,t}$ to be the probability we put arm $i$ into status $(u,a,t)$, conditioned on arriving at node $u$ after playing action $b$ on node $v$ at time $t'$.  The following lemma tells us that such $q$'s always exist, and that we can find them in polynomial time:
\begin{lemma}\label{flowdec}
	Suppose we are given the $x$'s of a feasible solution to $\mathtt{(PolyLP)}$.  Then we can find $\{q_{v,b,t',u,a,t}:u\in\cS\setminus\{\rho_1,\ldots,\rho_n\},a\in A,t\in[B],(v,b)\in\Par(u),t'<t\}$ in polynomial time such that
	\begin{subequations}
		\begin{align}
		\sum_{a\in A}\sum_{t>t'}q_{v,b,t',u,a,t} & \le1 && u\in\cS\setminus\{\rho_1,\ldots,\rho_n\},\ (v,b)\in\Par(u),\ t'\in[B-1] \label{flowdec_eqn1} \\
		\sum_{(v,b)\in\Par(u)}\sum_{t'<t}x^b_{v,t'}\cdot p^b_{v,u}\cdot q_{v,b,t',u,a,t} & =x^a_{u,t} && u\in\cS\setminus\{\rho_1,\ldots,\rho_n\},\ a\in A,\ t\in\{2,\ldots,B\} \label{flowdec_eqn2}
		\end{align}
	\end{subequations}
	Furthermore, if $u\in\cB$, then we can strengthen (\ref{flowdec_eqn1}) to $q_{v,b,t',u,\alpha,t'+1}=1$.
\end{lemma}
(\ref{flowdec_eqn1}) ensures that the probabilities telling us what to do, when we arrive at node $u$ after playing action $b$ on node $v$ at time $t'$, are well-defined; the case where $u$ is a bridge node will be needed to prove Theorem~\ref{mainresult3}.  For all $i\in[n]$, $u\in\cS_i\setminus\{\rho_i\}$, $(v,b)\in\Par(u)$, and $t'\in[B-1]$, define $q_{v,b,t',u,\infty}=1-\sum_{a\in A}\sum_{t>t'}q_{v,b,t',u,a,t}$, the probability we abandon arm $i$ after making the transition to $u$.  (\ref{flowdec_eqn2}) will be used in the analysis to provide a guarantee on each status $(u,a,t)$ ever being reached.

Lemma~\ref{flowdec} is the replacement for \emph{convex decomposition} from \cite{GKMR11}, and its proof is deferred to Appendix~C.  Having defined the $q$'s, the overall algorithm can now be described in two steps:
\begin{enumerate}
	\item While there exists an arm with priority not $\infty$, play an arm with the lowest priority (breaking ties arbitrarily) until it arrives at a status $(u,a,t)$ such that $t\ge 2\cdot\depth(u)$ ($t=\infty$ would suffice).
	\item Repeat until all arms have priority $\infty$.
\end{enumerate}
Of course, we are constrained by a budget of $B$ time steps, but it will simplify the analysis to assume our algorithm finishes all the arms and collects reward only for plays up to time $B$.  Under this assumption, the statuses an arm goes through is independent of the outcomes on all other arms; the inter-dependence only affects the order in which arms are played (and thus which nodes obtain reward).

Also, note that this is only a valid algorithm because Theorem~\ref{mainresult2} assumes all processing times are $1$, so there are no bridge nodes.  If there were bridge nodes, then we may not be allowed to switch to an arm with lowest priority index, being forced to play the arm on a bridge node.

\subsection{Analysis of algorithm.} For all $i\in[n]$, $u\in\cS_i$, $a\in A$, $t\in[B]$, let $\timesf(u,a,t)$ be the random variable for the time step at which our algorithm plays arm $i$ from status $(u,a,t)$, with $\timesf(u,a,t)=\infty$ if arm $i$ never gets in status $(u,a,t)$.  Then $\Pr[\timesf(\rho_i,a,t)<\infty]=\frac{x^a_{\rho_i,t}}{C}$ for all $i\in[n]$, $a\in A$, $t\in[B]$.  If $u$ is a non-root node, then we can induct on $\depth(u)$ to prove for all $a\in A$, $t\in[B]$ that
\begin{eqnarray}
\Pr[\timesf(u,a,t)<\infty] & = & \sum_{(v,b)\in\Par(u)}\sum_{t'<t}\Pr[\timesf(v,b,t')<\infty]\cdot p^b_{v,u}\cdot q_{v,b,t',u,a,t} \nonumber \\
& = & \sum_{(v,b)\in\Par(u)}\sum_{t'<t}\frac{x^b_{v,t'}}{C}\cdot p^b_{v,u}\cdot q_{v,b,t',u,a,t} \nonumber \\
& = & \frac{x^a_{u,t}}{C} \label{constant_factor}
\end{eqnarray}
where the final equality follows from Lemma~\ref{flowdec}.

For an event $\cA$, let $\bb_{\cA}$ be the indicator random variable for $\cA$.  The expected reward obtained by our algorithm is
\begin{eqnarray*}
	& & \mathbb{E}\Big[\sum_{u\in\cS}\sum_{a\in A}r^a_u\sum_{t=1}^B\bb_{\{\timesf(u,a,t)\le B\}}\Big] \\
	& = & \sum_{u\in\cS}\sum_{a\in A}r^a_u\sum_{t=1}^B \mathbb{E}[\bb_{\{\timesf(u,a,t)\le B\}}\ |\ \timesf(u,a,t)<\infty]\cdot\Pr[\timesf(u,a,t)<\infty] \\
	& = & \sum_{u\in\cS}\sum_{a\in A}r^a_u\sum_{t=1}^B \Pr[\timesf(u,a,t)\le B\ |\ \timesf(u,a,t)<\infty]\cdot\frac{x^a_{u,t}}{C}
\end{eqnarray*}

For the remainder of this subsection, we will set $C=3$ and prove for an arbitrary $i\in[n]$, $u\in\cS_i$, $a\in A$, $t\in[B]$ that $\Pr[\timesf(u,a,t)\le B\ |\ \timesf(u,a,t)<\infty]\ge\frac{4}{9}$.  It suffices to prove that $\Pr[\timesf(u,a,t)\le t\ |\ \timesf(u,a,t)<\infty]\ge\frac{4}{9}$, since $t\le B$.

\underline{Case 1}. Suppose $t\ge 2\cdot\depth(u)$.  We prove that conditioned on the event $\{\timesf(u,a,t)<\infty\}$, $\{\timesf(u,a,t)>t\}$ occurs with probability at most $\frac{5}{9}$.

Note that every node $v$ can have at most one $b,t'$ such that $\timesf(v,b,t')<\infty$; let $\timesf(v)$ denote this quantity (and be $\infty$ if $\timesf(v,b,t')=\infty$ for all $b\in A$, $t'\in[B]$).  The nodes $v$ that are played before $u$ are those with $\timesf(v)<\timesf(u,a,t)$.  Since our algorithm plays a node at every time step, $\timesf(u,a,t)>t$ if and only if there are $t$ or more nodes $v\neq u$ such that $\timesf(v)<\timesf(u,a,t)$.  But this is equivalent to there being exactly $t$ nodes $v\neq u$ such that $\timesf(v)<\timesf(u,a,t)$ and $\timesf(v)\le t$.  The $\depth(u)$ ancestors of $u$ are guaranteed to satisfy this.

Hence the event $\{\timesf(u,a,t)>t\}$ is equivalent to $\{\depth(u)+\sum_{v\in\cS\setminus\cS_i}\bb_{\{\timesf(v)<\timesf(u,a,t)\}}\cdot\bb_{\{\timesf(v)\le t\}}=t\}$.  But $t\ge 2\cdot\depth(u)$, so this implies $\{\sum_{v\in\cS\setminus\cS_i}\bb_{\{\timesf(v)<\timesf(u,a,t)\}}\cdot\bb_{\{\timesf(v)\le t\}}\ge\frac{t}{2}\}\implies\{\sum_{v\in\cS\setminus\cS_i}\bb_{\{\timesf(v)<\timesf(u,a,t)\}}\ge\frac{t}{2}\}$.  Now, whether the sum is at least $\frac{t}{2}$ is unchanged if we exclude all $v$ such that $\depth(v)\ge\frac{t}{2}$.  Indeed, if any such $v$ satisfies $\timesf(v)<\timesf(u,a,t)$, then all of its ancestors also do, and its first $\lceil\frac{t}{2}\rceil$ ancestors ensure that the sum, without any nodes of depth at least $\frac{t}{2}$, is at least $\frac{t}{2}$.  Thus, the last event is equivalent to
\begin{equation}\label{dagger3}
\big\{\sum_{v\in\cS\setminus\cS_i:\depth(v)<\frac{t}{2}}\bb_{\{\timesf(v)<\timesf(u,a,t)\}}\ge\frac{t}{2}\big\}
\end{equation}

Suppose $\timesf(v)=\timesf(v,b,t')$ for some $b\in A$ and $t'\in[B]$.  We would like to argue that in order for both $\timesf(v)<\timesf(u,a,t)$ and $\depth(v)<\frac{t}{2}$ to hold, it must be the case that $t'\le t$.  Suppose to the contrary that $t'>t$.  If $t'\ge 2\cdot\depth(v)$, then the algorithm can only play $(v,b,t')$ once $t'$ becomes the lowest priority index, which must happen after $(u,a,t)$ becomes the lowest priority index, hence $\timesf(v,b,t')<\timesf(u,a,t)$ is impossible.  Otherwise, if $t'<2\cdot\depth(v)$, then $\depth(v)>\frac{t'}{2}>\frac{t}{2}$, violating $\depth(v)<\frac{t}{2}$.  Thus indeed $t'\le t$ and
\begin{eqnarray*}
	(\ref{dagger3}) & \iff & \big\{\sum_{v\in\cS\setminus\cS_i:\depth(v)<\frac{t}{2}}\sum_{b\in A}\sum_{t'=1}^t\bb_{\{\timesf(v,b,t')<\timesf(u,a,t)\}}\ge\frac{t}{2}\big\} \\
	& \implies & \big\{\sum_{v\in\cS\setminus\cS_i}\sum_{b\in A}\sum_{t'=1}^t\bb_{\{\timesf(v,b,t')<\infty\}}\ge\frac{t}{2}\big\}
\end{eqnarray*}

We establish that the probability of interest $\Pr[\timesf(u,a,t)>t\ |\ \timesf(u,a,t)<\infty]$ is at most
\begin{eqnarray*}
	& & \Pr\Big[\sum_{v\in\cS\setminus\cS_i}\sum_{b\in A}\sum_{t'=1}^t\bb_{\{\timesf(v,b,t')<\infty\}}\ge\frac{t}{2}\ |\ \timesf(u,a,t)<\infty\Big] \\
	& = & \Pr\Big[\sum_{j\neq i}\sum_{v\in\cS_j}\sum_{b\in A}\sum_{t'=1}^t\bb_{\{\timesf(v,b,t')<\infty\}}\ge\frac{t}{2}\Big]
\end{eqnarray*}
where we remove the conditioning due to independence between arms.  Now, let
\begin{equation*}
Y_j=\min\big\{\sum_{v\in\cS_j}\sum_{b\in A}\sum_{t'=1}^t\bb_{\{\timesf(v,b,t')<\infty\}},\frac{t}{2}\big\}
\end{equation*}
for all $j\neq i$.  The previous probability is equal to $\Pr[\sum_{j\neq i}Y_j\ge\frac{t}{2}]$.  Note that
\begin{eqnarray*}
	\mathbb{E}\Big[\sum_{j\neq i}Y_j\Big] & \le & \sum_{j\neq i}\sum_{v\in\cS_j}\sum_{b\in A}\sum_{t'=1}^t\Pr[\timesf(v,b,t')<\infty] \\
	& \le & \sum_{t'=1}^t\sum_{v\in\cS}\sum_{b\in A}\frac{x^b_{v,t'}}{3} \\
	& \le & \frac{t}{3}
\end{eqnarray*}
where the second inequality uses (\ref{constant_factor}), and the final inequality uses (\ref{P1}).  We can do better than the Markov bound on $\Pr[\sum_{j\neq i}Y_j\ge\frac{t}{2}]$ because the random variables $\{Y_j\}_{j\neq i}$ are independent.  Furthermore, each $Y_j$ is non-zero with probability at most $\frac{1}{3}$ (arm $j$ is never touched with probability at least $\frac{2}{3}$), so since is at most $\frac{t}{2}$ when it is non-zero, $\mathbb{E}[Y_j]\le\frac{t}{6}$ for all $j\neq i$.  We now invoke the following lemma:
\begin{lemma}
	\label{grind}
	Let $t>0$ be arbitrary and $Y_1,\ldots,Y_m$ be independent non-negative random variables with individual expectations at most $\frac{t}{6}$ and sum of expectations at most $\frac{t}{3}$.  Then $\Pr[\sum_{j=1}^m Y_j\ge\frac{t}{2}]$ is maximized when only two random variables are non-zero, each taking value $\frac{t}{2}$ with probability $\frac{1}{3}$ (and value $0$ otherwise).  Therefore, $\Pr[\sum_{j=1}^m Y_j\ge\frac{t}{2}]\le 1-(1-\frac{1}{3})^2=\frac{5}{9}$.
\end{lemma}

This lemma would complete the proof that $\Pr[\timesf(u,a,t)>t\ |\ \timesf(u,a,t)<\infty]\le\frac{5}{9}$ under Case 1, where $t\ge2\cdot\depth(u)$.  We defer the proof of Lemma~\ref{grind} to Appendix~C.  It uses the conjecture of Samuels from \cite{Sam66} for $n=3$; the conjecture has been proven for $n\le 4$ in \cite{Sam68}.  The proof also uses a technical lemma of Bansal et al.\ from \cite{BGL+12}.

\underline{Case 2}. Suppose $t<2\cdot\depth(u)$.  Then $\depth(u)$ must be at least $1$, so conditioned on $\timesf(u,a,t)<\infty$, there must be some $(v,b)\in\Par(u)$ and $t'<t$ such that $\timesf(v,b,t')<\infty$.  Furthermore, the algorithm will play status $(u,a,t)$ at time step $\timesf(v,b,t')+1$, so $\timesf(u,a,t)\le t$ will hold so long as $\timesf(v,b,t')\le t'$, since $t'<t$.  Thus $\Pr[\timesf(u,a,t)\le t\ |\ \timesf(u,a,t)<\infty]\ge\Pr[\timesf(v,b,t')\le t\ |\ \timesf(v,b,t')<\infty]$.  We can iterate this argument until the the problem reduces to Case 1.  This completes the proof that $\Pr[\timesf(u,a,t)\le t\ |\ \timesf(u,a,t)<\infty]\ge\frac{4}{9}$ under Case 2.

\

Therefore, the expected reward obtained by our algorithm is at least $\sum_{u\in\cS}\sum_{a\in A}r^a_u\sum_{t=1}^B(1-\frac{5}{9})\frac{x^a_{u,t}}{3}$, which is the same as $\frac{4}{27}\OPT_{\mathtt{PolyLP}}$, completing the proof of Theorem~\ref{mainresult2}.

\subsection{Proof of Theorem~\ref{mainresult3}.}\label{proof_section2.5}

In this subsection we show how to modify the algorithm and analysis when there are multi-period actions, to prove Theorem~\ref{mainresult3}.  As mentioned in Subsection~\ref{descr_priority_alg}, we must modify Step~1 of the algorithm when there are bridge nodes.  If we arrive at a status $(u,a,t)$ such that $t\ge2\cdot\depth(u)$ but $u\in\cB$ (and $a=\alpha$), we are forced to immediately play the same arm again, instead of switching to another arm with a lower priority index.

The overall framework of the analysis still holds, except now the bound is optimized when we set $C=6$.  Our goal is to prove for an arbitrary $i\in[n]$, $u\in\cS_i$, $a\in A$, $t\in[B]$ that $\Pr[\timesf(u,a,t)>t\ |\ \timesf(u,a,t)<\infty]\le\frac{1}{2}$.  We still have that event $\{\timesf(u,a,t)>t\}$ implies (\ref{dagger3}).  Suppose for an arbitrary $v\in\cS\setminus\cS_i$ that $\depth(v)<\frac{t}{2}$ and $\timesf(v)<\timesf(u,a,t)$, where $\timesf(v)=\timesf(v,b,t')$.  We can no longer argue that $t'\le t$, but we would like to argue that $t'\le\frac{3t}{2}$.  Suppose to the contrary that $t'>\frac{3t}{2}$.

Then $t'>3\cdot\depth(v)$, so $t'\ge2\cdot\depth(v)$ ie.\ we would check priorities before playing $(v,b,t')$.  However, if $v\in\cB$, then it could be the case that $\timesf(v,b,t')<\timesf(u,a,t)$ even though $t'>t$.  If so, consider $w$, the youngest (largest depth) ancestor of $v$ that isn't a bridge node.  Suppose $\timesf(w)=\timesf(w,b',t'')$; it must be the case that $t''\le t$.  By the final statement of Lemma~\ref{flowdec}, the $\depth(v)-\depth(w)$ immediate descendents of $w$, which are bridge nodes, must have priority indices $t''+1,\ldots,t''+\depth(v)-\depth(w)$, respectively.  The youngest of these descendents is $v$, hence $t'=t''+\depth(v)-\depth(w)$.  But $t''\le t$ and $\depth(v)<\frac{t}{2}$, so $t'<\frac{3t}{2}$, causing a contradiction.

Therefore $t'\le\frac{3t}{2}$.  The bound on $\mathbb{E}[\sum_{j\neq i}Y_j]$ changes to
\begin{eqnarray*}
	\mathbb{E}\Big[\sum_{j\neq i}Y_j\Big] & \le & \sum_{j\neq i}\sum_{v\in\cS_j}\sum_{b\in A}\sum_{t'=1}^{\frac{3t}{2}}\Pr[\timesf(v,b,t')<\infty] \\
	& \le & \sum_{t'=1}^{\frac{3t}{2}}\sum_{v\in\cS}\sum_{b\in A}\frac{x^b_{v,t'}}{6} \\
	& \le & \frac{t}{4}
\end{eqnarray*}

Thus $\Pr[\timesf(u,a,t)>t\ |\ \timesf(u,a,t)<\infty]\le\Pr[\sum_{j\neq i}Y_j\ge\frac{t}{2}]\le\frac{1}{2}$ where the final inequality is Markov's inequality.  Note that we cannot use the stronger Samuels' conjecture here because we would need it for $n=5$, which is unproven; if we could, then we could get a better approximation factor (and we would re-optimize $C$).

The rest of the analysis, including Case~2, is the same as before.  Therefore, the expected reward obtained by our algorithm is at least $\sum_{u\in\cS}\sum_{a\in A}r^a_u\sum_{t=1}^B(1-\frac{1}{2})\frac{x^a_{u,t}}{6}$, which is the same as $\frac{1}{12}\OPT_{\mathtt{PolyLP}}$, completing the proof of Theorem~\ref{mainresult3}.

\section{Conclusion and open questions.}

In this paper, we presented a $(\frac{1}{2}-\varepsilon)$-approximation for the fully general \emph{MAB superprocess with multi-period actions---no preemption} problem by following a scaled copy of an optimal solution to the LP relaxation, and this is tight.  However, when preemption is allowed, we were only able to obtain a $\frac{1}{12}$-approximation, using the solution to the LP relaxation mainly for generating priorities, and resorting to weak Markov-type bounds in the analysis.  It seems difficult to follow a scaled copy of a solution to the LP relaxation when preemption is allowed, because arms can be paused and restarted.  We do conjecture that our bound of $(\frac{1}{2}+\varepsilon)$ on the gap of the LP is correct in this case and that it is possible to obtain a $(\frac{1}{2}-\varepsilon)$-approximation, but this remains an open problem.  Also, we have not explored how our techniques apply to certain extensions of the multi-armed bandit problem (switching costs, simultaneous plays, delayed feedback, contextual information, etc.).

\textbf{Acknowledgments.} The author's research was partly supported by a Natural Sciences and Engineering Research Council of Canada (NSERC) Postgraduate Scholarship - Doctoral.  The author would like to thank Michel Goemans for helpful discussions and suggestions on the presentation of the results.  A preliminary version of this paper appeared in the ACM-SIAM Symposium on Discrete Algorithms (SODA), 2014, and useful remarks from several anonymous referees there have also improved the presentation of the results.

\bibliographystyle{amsalpha}
\bibliography{bibliography}

\section{Appendix A: Example showing preemption is necessary for uncorrelated SK.}\label{appxA}

Consider the following example: there are $n=3$ items, $I_1,I_2,I_3$.  $I_1$ instantiates to size $6$ with probability $\frac{1}{2}$, and size $1$ with probability $\frac{1}{2}$.  $I_2$ deterministically instantiates to size $9$.  $I_3$ instantiates to size $8$ with probability $\frac{1}{2}$, and size $4$ with probability $\frac{1}{2}$.  $I_1,I_2,I_3$, if successfully inserted, return rewards of $4,9,8$, respectively.  We have a knapsack of size $10$.

We describe the optimal preempting policy.  First we insert $I_1$.  After $1$ unit of time, if $I_1$ completes, we go on to insert $I_2$, which will deterministically fit.  If $I_1$ doesn't complete, we set it aside and insert $I_3$ to completion.  If it instantiates to size $8$, then we cannot get any more reward from other items.  If it instantiates to size $4$, then we can go back and finish inserting the remaining $5$ units of $I_1$.  The expected reward of this policy is $\frac{1}{2}(4+9)+\frac{1}{2}(\frac{1}{2}8+\frac{1}{2}(8+4))=11.5$.

Now we enumerate the policies that can only cancel but not preempt.  If we first insert $I_2$, then the best we can do is try to insert $I_1$ afterward, getting a total expected reward of $11$.  Note that any policy never fitting $I_2$ can obtain reward at most $11$, since with probability $\frac{1}{4}$ we cannot fit both $I_1$ and $I_3$.  This rules out policies that start with $I_3$, which has no chance of fitting alongside $I_2$.  Remaining are the policies that first insert $I_1$.  If it doesn't complete after $1$ unit of time, then we can either settle for the $9$ reward of $I_2$, or finish processing $I_1$ with the hope of finishing $I_3$ afterward.  However, in this case, $I_3$ only finishes half the time, so we earn more expected reward by settling for $I_2$.  Therefore, the best we can do after first inserting $I_1$ is to stop processing it after time $1$ (regardless of whether it completes), and process $I_2$, earning a total expected reward of $11$.

We have shown that indeed, for uncorrelated SK, there is a gap between policies that can preempt versus policies that can only cancel.  It appears that this gap is bounded by a constant, contrary to the gap between policies that can cancel versus policies that cannot cancel (see \cite[appx.~A.1]{GKMR11}).

\section{Appendix B: Proofs from Section~\ref{preliminaries}.}\label{appxB}

\subsection{Proof of Lemma~\ref{LPproj}.} Suppose we are given $\{z^a_{\pi,i,t}\},\{y_{\pi,t}\}$ satisfying (\ref{E2a})-(\ref{E3}), (\ref{E4a})-(\ref{E5}) which imply (\ref{E1}).  For all $i\in[n]$, $u\in\cS_i$, $t\in[B]$, let $s_{u,t}=\sum_{\pi\in\bcS:\pi_i=u}y_{\pi,t}$, and let $x^a_{u,t}=\sum_{\pi\in\bcS:\pi_i=u}z^a_{\pi,i,t}$ for each $a\in A$.  We aim to show $\{x^a_{u,t}\},\{s_{u,t}\}$ satisfies (\ref{P2a})-(\ref{P3}), (\ref{P1}), (\ref{P4a})-(\ref{P5}) and makes (\ref{P0}) the same objective function as (\ref{E0}).  For convenience, we adopt the notation that $x_{u,t}=\sum_{a\in A}x^a_{u,t}$ and $z_{\pi,i,t}=\sum_{a\in A}z^a_{\pi,i,t}$.

(\ref{P1}): $\sum_{u\in\cS}x_{u,t}=\sum_{i=1}^n\sum_{u\in\cS_i}\sum_{\pi:\pi_i=u}z_{\pi,i,t}=\sum_{\pi\in\bcS}\sum_{i=1}^n\sum_{u\in\cS_i:u=\pi_i}z_{\pi,i,t}$.  But there is a unique $u\in\cS_i$ such that $u=\pi_i$, so the sum equals $\sum_{\pi\in\bcS}\sum_{i=1}^n z_{\pi,i,t}$, which is at most $1$ by (\ref{E1}).

(\ref{P2a}): For $u\in\cS_i$, $x_{u,t}=\sum_{\pi:\pi_i=u}z_{\pi,i,t}$, and each term in the sum is at most $y_{\pi,t}$ by (\ref{E2a}) and (\ref{E3}), hence $x_{u,t}\le\sum_{\pi:\pi_i=u}y_{\pi,t}=s_{u,t}$.

(\ref{P2b}): For $u\in\cB$, $x^{\alpha}_{u,t}=\sum_{\pi:\pi_i=u}z^{\alpha}_{\pi,i,t}$, and each term in the sum is equal to $y_{\pi,t}$ by (\ref{E2b}), hence $x^{\alpha}_{u,t}=\sum_{\pi:\pi_i=u}y_{\pi,t}=s_{u,t}$.

(\ref{P3}), (\ref{P4a}), and (\ref{P4b}) are immediate from (\ref{E3}), (\ref{E4a}), and (\ref{E4b}), respectively.  For (\ref{P5}), fix $t>1$, $i\in[n]$, and $u\in\cS_i$.  Sum (\ref{E5}) over $\{\pi:\pi_i=u\}$ to get
\begin{eqnarray*}
	\sum_{\pi:\pi_i=u}y_{\pi,t} & = & \sum_{\pi:\pi_i=u}y_{\pi,t-1}-\sum_{\pi:\pi_i=u}\sum_{j=1}^n z_{\pi,j,t-1}+\sum_{\pi:\pi_i=u}\sum_{j=1}^n\sum_{(v,a)\in\Par(\pi_j)}z^a_{\pi^v,j,t-1}\cdot p^a_{v,\pi_j} \\
	s_{u,t} & = & s_{u,t-1}-\sum_{\pi:\pi_i=u}z_{\pi,i,t-1}-\sum_{\pi:\pi_i=u}\sum_{j\neq i}z_{\pi,j,t-1} \\
	& & +\sum_{\pi:\pi_i=u}\sum_{(v,a)\in\Par(u)}z^a_{\pi^v,i,t-1}\cdot p^a_{v,u}+\sum_{\pi:\pi_i=u}\sum_{j\neq i}\sum_{(v,a)\in\Par(\pi_j)}z^a_{\pi^v,j,t-1}\cdot p^a_{v,\pi_j} \\
	s_{u,t} & = & s_{u,t-1}-x_{u,t-1}-\sum_{\pi:\pi_i=u}\sum_{j\neq i}z_{\pi,j,t-1} \\
	& & +\sum_{(v,a)\in\Par(u)}(\sum_{\pi:\pi_i=u}z^a_{\pi^v,i,t-1})\cdot p^a_{v,u}+\sum_{j\neq i}\sum_{v\in\cS_j}\sum_{a\in A}\sum_{\{\pi:\pi_i=u,\Par(\pi_j)\ni(v,a)\}}z^a_{\pi^v,j,t-1}\cdot p^a_{v,\pi_j} \\
	s_{u,t} & = & s_{u,t-1}-x_{u,t-1}-\sum_{\pi:\pi_i=u}\sum_{j\neq i}z_{\pi,j,t-1} \\
	& & +\sum_{(v,a)\in\Par(u)}(\sum_{\pi:\pi_i=v}z^a_{\pi,i,t-1})\cdot p^a_{v,u}+\sum_{j\neq i}\sum_{v\in\cS_j}\sum_{a\in A}\sum_{\pi:\pi_i=u,\pi_j=v}z^a_{\pi,j,t-1}\cdot(\sum_{w:p^a_{v,w}>0}p^a_{v,w}) \\
	s_{u,t} & = & s_{u,t-1}-x_{u,t-1}-\sum_{\pi:\pi_i=u}\sum_{j\neq i}z_{\pi,j,t-1} \\
	& & +\sum_{(v,a)\in\Par(u)}x^a_{v,t-1}\cdot p^a_{v,u}+\sum_{j\neq i}\sum_{v\in\cS_j}\sum_{a\in A}\sum_{\pi:\pi_i=u,\pi_j=v}z^a_{\pi,j,t-1}\cdot(1) \\
	s_{u,t} & = & s_{u,t-1}-x_{u,t-1}-\sum_{j\neq i}\sum_{\pi:\pi_i=u}z_{\pi,j,t-1}+\sum_{(v,a)\in\Par(u)}x^a_{v,t-1}\cdot p^a_{v,u}+\sum_{j\neq i}\sum_{\pi:\pi_i=u}z_{\pi,j,t-1} \\
	s_{u,t} & = & s_{u,t-1}-x_{u,t-1}+\sum_{(v,a)\in\Par(u)}x^a_{v,t-1}\cdot p^a_{v,u} \\
\end{eqnarray*}
which is exactly (\ref{P5}).

(\ref{P0}): $\sum_{u\in\cS}\sum_{a\in A}r^a_u\sum_{t=1}^B x^a_{u,t}=\sum_{i=1}^n\sum_{u\in\cS_i}\sum_{a\in A}r^a_u\sum_{t=1}^B\sum_{\pi:\pi_i=u}z^a_{\pi,i,t}$, and by the same manipulation we made for (\ref{P1}), this is equal to $\sum_{\pi\in\bcS}\sum_{i=1}^n\sum_{a\in A}r^a_{\pi_i}\sum_{t=1}^B z^a_{\pi,i,t}$.  Thus (\ref{P0}) is the same as (\ref{E0}), completing the proof of Lemma~\ref{LPproj}.

\subsection{Proof of Lemma~\ref{LPprimeproj}.} Suppose we are given $\{z^a_{\pi,i,t}\},\{y_{\pi,t}\}$ satisfying (\ref{E2a'})-(\ref{E3'}), (\ref{E4a'})-(\ref{E5c'}) which imply (\ref{E1'}).  For all $i\in[n]$, $u\in\cS_i$, $t\in[B]$, let $s_{u,t}=\sum_{\pi\in\bcS':\pi_i=u}y_{\pi,t}$, and let $x^a_{u,t}=\sum_{\pi\in\bcS':\pi_i=u}z^a_{\pi,i,t}$ for each $a\in A$.  We aim to show $\{x^a_{u,t}\},\{s_{u,t}\}$ satisfies (\ref{P2a})-(\ref{P3}), (\ref{P1}), (\ref{P4a'})-(\ref{P5b'}) and makes (\ref{P0}) the same objective function as (\ref{E0'}).  For convenience, we adopt the notation that $x_{u,t}=\sum_{a\in A}x^a_{u,t}$ and $z_{\pi,i,t}=\sum_{a\in A}z^a_{\pi,i,t}$.

(\ref{P1}): $\sum_{u\in\cS}x_{u,t}=\sum_{i=1}^n\sum_{u\in\cS_i}\sum_{\pi:\pi_i=u}z_{\pi,i,t}=\sum_{\pi\in\bcS'}\sum_{i=1}^n\sum_{u\in\cS_i:u=\pi_i}z_{\pi,i,t}$.  The difference from the previous derivation of (\ref{P1}) is that there is only a unique $u\in\cS_i$ such that $u=\pi_i$, if $\pi_i\neq\phi_i$.  So the sum equals $\sum_{\pi\in\bcS'}\sum_{i\in I(\pi)}z_{\pi,i,t}$, which is at most $1$ by (\ref{E1'}).

Using this same manipulation, the equivalence of (\ref{P0}) and (\ref{E0'}) follows the same derivation as before.  (\ref{P2a}) and (\ref{P2b}) also follow the same derivations as before; (\ref{P3}), (\ref{P4a'}), and (\ref{P4b'}) are immediate.  It remains to prove (\ref{P5a'}) and (\ref{P5b'}).

(\ref{P5b'}): Fix $t>1$, $i\in[n]$, and $u\in\cS_i\setminus\{\rho_i\}$.  First consider the case where $\depth(u)>1$.  All $\pi\in\bcS'$ such that $\pi_i=u$ fall under (\ref{E5c'}), so we can sum over these $\pi$ to get
\begin{eqnarray*}
	\sum_{\pi:\pi_i=u}y_{\pi,t} & = & \sum_{\pi:\pi_i=u}\sum_{(v,a)\in\Par(u)}z^a_{\pi^v,i,t-1}\cdot p^a_{v,u} \\
	s_{u,t} & = & \sum_{(v,a)\in\Par(u)}(\sum_{\pi:\pi_i=u}z^a_{\pi^v,i,t-1})\cdot p^a_{v,u} \\
\end{eqnarray*}
Since $\depth(u)>1$, $v\neq\rho_i$, so $\{\pi^v:\pi\in\bcS',\pi_i=u\}=\{\pi:\pi\in\bcS',\pi_i=v\}$.  Hence the RHS of the above equals $\sum_{(v,a)\in\Par(u)}x^a_{v,t-1}\cdot p^a_{v,u}$ which is exactly (\ref{P5b'}).

For the other case where $\depth(u)=1$, all $\pi\in\bcS'$ such that $\pi_i=u$ fall under (\ref{E5b'}), so we can sum over these $\pi$ to get
\begin{eqnarray*}
	\sum_{\pi:\pi_i=u}y_{\pi,t} & = & \sum_{\pi:\pi_i=u}\sum_{a:(\rho_i,a)\in\Par(u)}(\sum_{\pi'\in\bcP(\pi^{\rho_i})}z^a_{\pi',i,t-1})\cdot p^a_{\rho_i,u} \\
	s_{u,t} & = & \sum_{a:(\rho_i,a)\in\Par(u)}(\sum_{\pi:\pi_i=u}\sum_{\pi'\in\bcP(\pi^{\rho_i})}z^a_{\pi',i,t-1})\cdot p^a_{\rho_i,u} \\
	s_{u,t} & = & \sum_{a:(\rho_i,a)\in\Par(u)}(\sum_{\pi:\pi_i=\rho_i}z^a_{\pi,i,t-1})\cdot p^a_{\rho_i,u} \\
	s_{u,t} & = & \sum_{a:(\rho_i,a)\in\Par(u)}x^a_{\rho_i,t-1}\cdot p^a_{\rho_i,u} \\
\end{eqnarray*}
We explain the third equality.  Since $u\neq\rho_i$ implies arm $i$ is the active arm in all of $\{\pi\in\bcS':\pi_i=u\}$, this set is equal to $\{\rho_1,\phi_1\}\times\cdots\times u\times\cdots\times\{\rho_n,\phi_n\}$.  Thus $\{\pi'\in\bcP(\pi^{\rho_i}):\pi\in\bcS',\pi_i=u\}=\{\pi'\in\bcP(\pi):\pi\in\{\rho_1,\phi_1\}\times\cdots\times\rho_i\times\cdots\times\{\rho_n,\phi_n\}\}$.  Recall that $\bcP(\pi)$ is the set of joint nodes that would transition to $\pi$ with no play.  Therefore, this set is equal to $\{\pi'\in\bcS':\pi'_i=\rho_i\}$, as desired.

(\ref{P5a'}): Fix $t>1$ and $i\in[n]$.  Unfortunately, $\pi\in\bcS'$ such that $\pi_i=\rho_i$ can fall under (\ref{E5a'}), (\ref{E5b'}), or (\ref{E5c'}).  First let's sum over the $\pi$ falling under (\ref{E5a'}):
\begin{eqnarray*}
	\sum_{\pi\notin\bcA:\pi_i=\rho_i}y_{\pi,t} & = & \sum_{\pi\notin\bcA:\pi_i=\rho_i}\sum_{\pi'\in\bcP(\pi)}\Big(y_{\pi',t-1}-\sum_{j\in I(\pi')}z_{\pi',j,t-1}\Big) \\
	& = & \sum_{\pi\in\bcS':\pi_i=\rho_i}\Big(y_{\pi,t-1}-z_{\pi,i,t-1}-\sum_{j\in I(\pi)\setminus\{i\}}z_{\pi,j,t-1}\Big) \\
	& = & s_{\rho_i,t-1}-x_{\rho_i,t-1}-\sum_{\pi\in\bcS':\pi_i=\rho_i}\sum_{j\in I(\pi)\setminus\{i\}}z_{\pi,j,t-1} \\
\end{eqnarray*}
where the second equality requires the same set bijection explained above.  Furthermore,
\begin{eqnarray*}
	\sum_{\pi\in\bcS':\pi_i=\rho_i}\sum_{j\in I(\pi)\setminus\{i\}}z_{\pi,j,t-1} & = & \sum_{k\neq i}\sum_{\pi\in\bcA_k:\pi_i=\rho_i}\Big(z_{\pi,k,t-1}+\sum_{j\in I(\pi)\setminus\{i,k\}}z_{\pi,j,t-1}\Big)+\sum_{\pi\notin\bcA:\pi_i=\rho_i}\sum_{j\in I(\pi)\setminus\{i\}}z_{\pi,j,t-1} \\
	& = & \sum_{k\neq i}\sum_{\pi\in\bcA_k:\pi_i=\rho_i}\Big(z_{\pi,k,t-1}+\sum_{j:\pi_j=\rho_j,j\neq i}z_{\pi,j,t-1}\Big)+\sum_{\pi\notin\bcA:\pi_i=\rho_i}\sum_{j:\pi_j=\rho_j,j\neq i}z_{\pi,j,t-1} \\
	& = & \sum_{k\neq i}\sum_{\pi\in\bcA_k:\pi_i=\rho_i}z_{\pi,k,t-1}+\sum_{\pi\in\bcS':\pi_i=\rho_i}\sum_{j:\pi_j=\rho_j,j\neq i}z_{\pi,j,t-1} \\
\end{eqnarray*}
Now let's sum over the $\pi$ falling under (\ref{E5c'}):
\begin{eqnarray*}
	\sum_{j\neq i}\sum_{\{\pi:\pi_i=\rho_i,\depth(\pi_j)>1\}}y_{\pi,t} & = & \sum_{j\neq i}\sum_{\{\pi:\pi_i=\rho_i,\depth(\pi_j)>1\}}\sum_{(v,a)\in\Par(\pi_j)}z^a_{\pi^v,j,t-1}\cdot p^a_{v,\pi_j} \\
	& = & \sum_{j\neq i}\sum_{v\in\cS_j\setminus\{\rho_j\}}\sum_{a\in A}\sum_{\{\pi:\pi_i=\rho_i,\Par(\pi_j)\ni(v,a)\}}z^a_{\pi^v,j,t-1}\cdot p^a_{v,\pi_j} \\
	& = & \sum_{j\neq i}\sum_{v\in\cS_j\setminus\{\rho_j\}}\sum_{a\in A}\sum_{\pi:\pi_i=\rho_i,\pi_j=v}z^a_{\pi,j,t-1}\cdot(\sum_{w:p^a_{v,w}>0}p^a_{v,w}) \\
	& = & \sum_{j\neq i}\sum_{v\in\cS_j\setminus\{\rho_j\}}\sum_{a\in A}\sum_{\pi:\pi_i=\rho_i,\pi_j=v}z^a_{\pi,j,t-1}\cdot(1) \\
	& = & \sum_{j\neq i}\sum_{\pi\in\bcA_j:\pi_i=\rho_i}z_{\pi,j,t-1} \\
\end{eqnarray*}
where the third equality uses the fact that $v\neq\rho_j$ to convert $\pi^v$ to $\pi$.  Finally, let's sum over the $\pi$ falling under (\ref{E5b'}):
\begin{eqnarray*}
	& & \sum_{j\neq i}\sum_{\{\pi:\pi_i=\rho_i,\depth(\pi_j)=1\}}y_{\pi,t} \\
	& = & \sum_{j\neq i}\sum_{\{\pi:\pi_i=\rho_i,\depth(\pi_j)=1\}}\sum_{a:(\rho_j,a)\in\Par(\pi_j)}(\sum_{\pi'\in\bcP(\pi^{\rho_j})}z^a_{\pi',j,t-1})\cdot p^a_{\rho_j,\pi_j} \\
	& = & \sum_{j\neq i}\sum_{a\in A}\sum_{\{\pi:\pi_i=\rho_i,\Par(\pi_j)\ni(\rho_j,a)\}}(\sum_{\pi'\in\bcP(\pi^{\rho_j})}z^a_{\pi',j,t-1})\cdot p^a_{\rho_j,\pi_j} \\
	& = & \sum_{j\neq i}\sum_{a\in A}\sum_{\pi:\pi_i=\rho_i,\pi_j=\rho_j}z^a_{\pi,j,t-1}\cdot(\sum_{w:p^a_{\rho_j,w}>0}p^a_{\rho_j,w}) \\
	& = & \sum_{j\neq i}\sum_{a\in A}\sum_{\pi:\pi_i=\rho_i,\pi_j=\rho_j}z^a_{\pi,j,t-1}\cdot(1) \\
	& = & \sum_{\pi:\pi_i=\rho_i}\sum_{j:\pi_j=\rho_j,j\neq i}z_{\pi,j,t-1} \\
\end{eqnarray*}
where the third equality requires the same set bijection again.  Combining the last four blocks of equations, we get $s_{\rho_i,t}=s_{\rho_i,t-1}-x_{\rho_i,t-1}$ which is exactly (\ref{P5a'}), completing the proof of Lemma~\ref{LPprimeproj}.

\section{Appendix C: Proofs from Section~\ref{proof_section2}.}\label{appxC}

\subsection{Proof of Lemma~\ref{flowdec}.}

Finding the $q$'s is a separate problem for each arm, so we can fix $i\in[n]$.  Furthermore, we can fix $u\in\cS_i\setminus\{\rho_i\}$; we will specify an algorithm that defines $\{q_{v,b,t',u,a,t}:a\in A,t\in[B],(v,b)\in\Par(u),t'<t\}$ satisfying (\ref{flowdec_eqn1}) and (\ref{flowdec_eqn2}).

Observe that by substituting (\ref{P2a}) into (\ref{P5}), we get $s_{u,t'}\le\sum_{(v,b)\in\Par(u)}x^b_{v,t'-1}\cdot p^b_{v,u}$ for all $t'>1$.  Summing over $t'=2,\ldots,t$ for an arbitrary $t\in[B]$, and using (\ref{P2a}) again on the LHS, we get $\sum_{t'=2}^t\sum_{a\in A}x^a_{u,t'}\le\sum_{t'=1}^{t-1}\sum_{(v,b)\in\Par(u)}x^b_{v,t'}\cdot p^b_{v,u}$.

Now, for all $t'=2,\ldots,B$ and $a\in A$, initialize $\tx^a_{t'}:=x^a_{u,t'}$.  For all $t'=1,\ldots,B-1$ and $(v,b)\in\Par(u)$, initialize $\tx^b_{v,t'}:=x^b_{v,t'}\cdot p^b_{v,u}$.  We are omitting the subscript $u$ because $u$ is fixed.  The following $B-1$ inequalities hold:
\begin{align}
\sum_{t'=2}^{t''}\sum_{a\in A}\tx^a_{t'} & \le\sum_{t'=1}^{t''-1}\sum_{(v,b)\in\Par(u)}\tx^b_{v,t'} && t''=2,\ldots,B \label{invariant}
\end{align}
The algorithm updates the variables $\tx^a_{t'}$ and $\tx^b_{v,t'}$ over iterations $t=2,\ldots,B$, but we will inductively show that inequality $t''$ of (\ref{invariant}) holds until the end of iteration $t''$.  The algorithm can be described as follows:

\

\noindent $\mathsf{Decomposition\ Algorithm}$
\begin{itemize}
	\item Initialize all $q_{v,b,t',u,a,t}:=0$.
	\item For $t=2,\ldots,B$:
	\begin{itemize}
		\item While there exists some $a\in A$ such that $\tx^a_t>0$:
		\begin{enumerate}
			\item Choose any non-zero $\tx^b_{v,t'}$, where $(v,b)\in\Par(u)$, with $t'<t$.
			\item Let $Q=\min\{\tx^a_t,\tx^b_{v,t'}\}$.
			\item Set $q_{v,b,t',u,a,t}:=\frac{Q}{x^b_{v,t'}\cdot p^b_{v,u}}$.
			\item Subtract Q from both $\tx^a_t$ and $\tx^b_{v,t'}$.
		\end{enumerate}
	\end{itemize}
\end{itemize}

\

Let's consider iteration $t$ of the algorithm.  The inequality of (\ref{invariant}) with $t''=t$ guarantees that there always exists such a non-zero $\tx^b_{v,t'}$ in Step~1.  In Step~4, Q is subtracted from both the LHS and RHS of all inequalities of (\ref{invariant}) with $t''\ge t$, so these inequalities continue to hold. (Q is also subtracted from the RHS of inequalities of (\ref{invariant}) with $t'<t''<t$, so these inequalities might cease to hold.) This inductively establishes that all inequalities of (\ref{invariant}) with $t''\ge t$ hold during iteration $t$, and thus Step~1 of the algorithm is well-defined.

Now we show that (\ref{flowdec_eqn2}) is satisfied.  Suppose on iteration $t$ of the algorithm, we have some $\tx^a_t>0$ and $\tx^b_{v,t'}>0$ on Step~1.  Note that $q_{v,b,t',u,a,t}$ must currently be $0$, since if it was already set, then either $\tx^a_t$ or $\tx^b_{v,t'}$ would have been reduced to $0$.  Therefore, in Step~3 we are incrementing the LHS of (\ref{flowdec_eqn2}) by $x^b_{v,t'}\cdot p^b_{v,u}\cdot\frac{Q}{x^b_{v,t'}\cdot p^b_{v,u}}=Q$, after which we are subtracting Q from $\tx^a_t$ in Step~4.  Since over iterations $t=2,\ldots,B$, for every $a\in A$, $\tx^a_t$ gets reduced from $x^a_{u,t}$ to $0$, it must be the case that every equation in (\ref{flowdec_eqn2}) holds by the end of the algorithm.

For (\ref{flowdec_eqn1}), we use a similar argument.  Fix some $(v,b)\in\Par(u)$ and $t'\in[B-1]$.  Whenever we add $\frac{Q}{x^b_{v,t'}\cdot p^b_{v,u}}$ to the LHS of (\ref{flowdec_eqn1}), we are reducing $\tx^b_{v,t'}$ by $Q$.  Since $\tx^b_{v,t'}$ starts at $x^b_{v,t'}\cdot p^b_{v,u}$ and cannot be reduced below $0$, the biggest we can make the LHS of (\ref{flowdec_eqn1}) is $\frac{x^b_{v,t'}\cdot p^b_{v,u}}{x^b_{v,t'}\cdot p^b_{v,u}}=1$.

Finally, it is clear that the algorithm takes polynomial time, since every time we loop through Steps~1 to 4 either $\tx^a_t$ or $\tx^b_{v,t'}$ goes from non-zero to zero, and there were only a polynomial number of such variables to begin with.  Other than the statement for $u\in\cB$, this completes the proof of Lemma~\ref{flowdec}.

Now, if $u\in\cB$, then we can strengthen (\ref{invariant}).  Indeed, substituting (\ref{P2b}) (instead of (\ref{P2a})) into (\ref{P5}), we get that all inequalities of (\ref{invariant}) hold as equality.  At the start of iteration $t=2$, it is the case that $\tx^{\alpha}_t=\sum_{(v,b)\in\Par(u)}\tx^b_{v,t-1}$, and by the end of the iteration, it will be the case that $\tx^{\alpha}_t=0$, and $\tx^b_{v,t-1}=0$, $q_{v,b,t-1,u,\alpha,t}=1$ for all $(v,b)\in\Par(u)$.  As a result, the equalities of (\ref{invariant}) with $t''>t$ will continue to hold as equality.  We can inductively apply this argument to establish that $q_{v,b,t',u,\alpha,t'+1}=1$ for all $(v,b)\in\Par(u)$ and $t'\in[B-1]$, as desired.

\subsection{Proof of Lemma~\ref{grind}.}

We make use of the following conjecture of Samuels, which is proven for $n\le 4$ (see \cite{Sam66,Sam68}):
\begin{conjecture}
	\label{samconj}
	Let $X_1,\ldots,X_n$ be independent non-negative random variables with respective expectations $\mu_1\ge\ldots\ge\mu_n$, and let $\lambda>\sum_{i=1}^n\mu_i$.  Then $\Pr[\sum_{i=1}^n X_i\ge\lambda]$ is maximized when the $X_i$'s are distributed as follows, for some index $k\in[n]$:
	\begin{itemize}
		\item For $i>k$, $X_i=\mu_i$ with probability $1$.
		\item For $i\le k$, $X_i=\lambda-\sum_{\ell=k+1}^n\mu_{\ell}$ with probability $\frac{\mu_i}{\lambda-\sum_{\ell=k+1}^n\mu_{\ell}}$, and $X_i=0$ otherwise.
	\end{itemize}
\end{conjecture}

If we have $\mathbb{E}[Y_i]+\mathbb{E}[Y_j]\le\frac{t}{6}$ for $i\neq j$, then we can treat $Y_i+Y_j$ as a single random variable satisfying $\mathbb{E}[Y_i+Y_j]\le\frac{t}{6}$.  By the pigeonhole principle, we can repeat this process until $n\le 3$, since $\sum_{j=1}^m\mathbb{E}[Y_j]\le\frac{t}{3}$.  In fact, we assume $n$ is exactly $3$ (we can add random variables that take constant value $0$ if necessary), so that we can apply Conjecture~\ref{samconj} for $n=3$, which has been proven to be true.  We get that $\Pr[Y_1+Y_2+Y_3\ge\frac{t}{2}]$ cannot exceed the maximum of the following (corresponding to the cases $k=3,2,1$, respectively):
\begin{itemize}
	\item $\displaystyle 1-\big(1-\frac{\mu_1}{\frac{t}{2}}\big)\big(1-\frac{\mu_2}{\frac{t}{2}}\big)\big(1-\frac{\mu_3}{\frac{t}{2}}\big)$
	\item $\displaystyle 1-\big(1-\frac{\mu_1}{\frac{t}{2}-\mu_3}\big)\big(1-\frac{\mu_2}{\frac{t}{2}-\mu_3}\big)$
	\item $\displaystyle 1-\big(1-\frac{\mu_1}{\frac{t}{2}-\mu_2-\mu_3}\big)$
\end{itemize}

Now we employ Lemma~4 from \cite{BGL+12} to bound these quantities:
\begin{lemma}
	Let $r$ and $p_{\max}$ be positive real values.  Consider the problem of maximizing $1-\prod_{i=1}^t(1-p_i)$ subject to the constraints $\sum_{i=1}^t p_i\le r$, and $0\le p_i\le p_{\max}$ for all $i$.  Denote the maximum value by $\beta(r, p_{\max})$.  Then
	\begin{eqnarray*}
		\beta(r,p_{\max}) & = & 1-(1-p_{\max})^{\lfloor\frac{r}{p_{\max}}\rfloor}(1-(r-\lfloor\frac{r}{p_{\max}}\rfloor\cdot p_{\max})) \\
		& \le & 1-(1-p_{\max})^{\frac{r}{p_{\max}}} \\
	\end{eqnarray*}
\end{lemma}

Recall that $\mu_1,\mu_2,\mu_3\le\frac{t}{6}$ and $\mu_1+\mu_2+\mu_3\le\frac{t}{3}$.
\begin{itemize}
	\item In the first case $k=3$, we get $p_{\max}=\frac{1}{3}$ and $r=\frac{2}{3}$, so the quantity is at most $\beta(\frac{2}{3},\frac{1}{3})\le 1-(1-\frac{1}{3})^2=\frac{5}{9}$, as desired.
	\item In the second case $k=2$, for an arbitrary $\mu_3\in[0,\frac{t}{6}]$, we get that the quantity is at most $\displaystyle\beta\big(\frac{\frac{t}{3}-\mu_3}{\frac{t}{2}-\mu_3},\frac{\frac{t}{6}}{\frac{t}{2}-\mu_3}\big)\le 1-\big(1-\frac{\frac{t}{6}}{\frac{t}{2}-\mu_3}\big)^{(\frac{t}{3}-\mu_3)/(\frac{t}{6})}$.  It can be checked that the maximum occurs at $\mu_3=0$, so the quantity is at most $\frac{5}{9}$ for any value of $\mu_3\in[0,\frac{t}{6}]$, as desired.
	\item In the third case $k=1$, we get that the quantity is at most $\displaystyle\frac{\mu_1}{\frac{t}{2}-(\frac{t}{3}-\mu_1)}=\frac{\mu_1}{\frac{t}{6}+\mu_1}$, which at most $\frac{1}{2}$ over $\mu_1\in[0,\frac{t}{6}]$, as desired.
\end{itemize}

Therefore, Conjecture~\ref{samconj} tells us that the maximum value of $\Pr[\sum_{j=1}^m Y_j\ge\frac{t}{2}]$ is $\frac{5}{9}$, completing the proof of Lemma~\ref{grind}.

\end{document}